\documentclass[sigconf]{acmart}
\acmConference[ICSE 2024]{The 46th ACM International Conference on Software Engineering}{14 - 20 April, 2024}{Lisbon, Portugal}

\usepackage{tikz}
\usetikzlibrary{arrows.meta,
                chains,
                positioning,
                shapes.geometric,
                shapes.symbols,
                decorations.pathreplacing,
                calligraphy
                }

\usepackage{amsmath,amsfonts,amsthm}
\usepackage{algpseudocode,algorithm}
\usepackage{tabularx,colortbl}
\usepackage{adjustbox}
\usepackage{comment}
\usepackage{pifont}
\usepackage{fancybox,graphicx}
\usepackage{textcomp}
\usepackage{booktabs}
\usepackage{diagbox}
\usepackage{subcaption}
\usepackage{multirow, multicol}
\usepackage{soul}
\usepackage{hyperref}
\usepackage{xspace}
\usepackage{makecell}
\usepackage{enumitem}

\newcommand{\RN}[1]{%
  \textup{\uppercase\expandafter{\romannumeral#1}}%
}
\usepackage{fontawesome}
\usepackage{xcolor}
\usepackage{array,colortbl,xcolor}

\usepackage{scalerel}

\def\BibTeX{{\rm B\kern-.05em{\sc i\kern-.025em b}\kern-.08em
    T\kern-.1667em\lower.7ex\hbox{E}\kern-.125emX}}
    
\usepackage{tcolorbox}

\definecolor{darkspringgreen}{rgb}{0.09, 0.45, 0.27}

\newcommand{\red}[1]{#1}

\definecolor{ao}{rgb}{0.12, 0.3, 0.17}

\newcommand{\boldparagraph}[1]{\vskip 0.05in\noindent\textbf{#1.}}

\newcommand{\italicparagraph}[1]{\vskip 0.05in\noindent\textit{(#1).}}

\newtheorem{theorem}{Theorem}
\newtheorem{lemma}{Lemma}
\newtheorem{sop}{Sketch of Proof}
\newtheorem{corollary}{Corollary}

\theoremstyle{definition}
\newtheorem{definition}{Definition}
\newtheorem{remark}{Remark}
\newtheorem{example}{Example}

\definecolor{turquoisegreen}{rgb}{0.63, 0.84, 0.71}
\definecolor{tuftsblue}{rgb}{0.28, 0.57, 0.81}
\definecolor{atomictangerine}{rgb}{1.0, 0.6, 0.4}
\definecolor{lightapricot}{rgb}{0.99, 0.84, 0.69}
\definecolor{babypink}{rgb}{0.96, 0.76, 0.76}
\definecolor{mistyrose}{rgb}{1.0, 0.89, 0.88}
\definecolor{moccasin}{rgb}{0.98, 0.92, 0.84}
\definecolor{teagreen}{rgb}{0.82, 0.94, 0.75}
\definecolor{lavender}{rgb}{0.9, 0.9, 0.98}
\definecolor{mygray}{rgb}{0.5,0.5,0.5}
\definecolor{mymauve}{rgb}{0.1,0.2,0.7}
\definecolor{olivegreen}{cmyk}{.6,.4,0.8,0}
\definecolor{eclipse}{RGB}{127,0,85}

\sethlcolor{lightapricot}

\makeatletter
\algrenewcommand\ALG@beginalgorithmic{\scriptsize}
\makeatother

\newcommand{\dsl}{\texttt{\small SLEEC DSL}\xspace}

\newcommand{\CRRI}{WFI\xspace}
\newcommand{\upCRRI}{Well-formedness\xspace}

\newcommand{\fol}{FOL$^*$\xspace}
\newcommand{\lego}{\texttt{LEGOS}\xspace}
\newcommand{\tool}{{\texttt{LEGOS-SLEEC}}\xspace}
\newcommand{\approach}{{N-Check}\xspace}
\newcommand{\quoted}[1]{``{#1}''}

\newcommand{\sig}{S}

\newcommand{\nNFR}{N-NFR}
\newcommand{\nNFRs}{N-NFRs}

\newcommand{\andS}{\textbf{and}}
\newcommand{\orS}{\textbf{or}}
\newcommand{\notS}{\textbf{not}\xspace}

\newcommand{\translate}[1]{T(#1)}
\newcommand{\sctranslate}[1]{\textsc{TSC}(#1)}
\newcommand{\timedtranslate}[1]{\textsc{T}_{\downarrow}(#1)}
\newcommand{\timedtranslatestar}[1]{\textsc{T}_{\downarrow}^{*}(#1)}

\newcommand{\within}[2]{#1\, \textbf{within} \,#2}
\newcommand{\rulesyntax}[2]{\textbf{when}\, #1 \, \textbf{then}  \,#2}
\newcommand{\factsyntax}[2]{\textbf{exists}\, #1 \, \textbf{while}  \,#2}
\newcommand{\lang}[1]{\mathcal{L}(#1)}

\newcommand{\noncompliant}[1]{\textsc{Noncomp}(#1)}
\newcommand{\triggerfunc}[1]{\textsc{Trig}(#1)}

\newcommand{\concernFact}{\textit{c}}
\newcommand{\purposeFact}{\textit{p}}

\newcommand{\normalize}{\textsc{Norm}\xspace}

\newcommand{\Class}[1]{C^{#1}}
\newcommand{\obj}[1]{o^{#1}}
\newcommand{\forced}[3]{\textsc{Forced}(#1, #2, #3)}
\newcommand{\blocked}[3]{\textsc{Blocked}(#1, #2, #3)}
\newcommand{\intblocked}[3]{\textsc{\_Blocked}(#1, #2, #3)}
\newcommand{\triggered}[2]{\textsc{Triggered}(#1, #2)}
\newcommand{\fullfilled}[3]{\textsc{Fulfilled}(#1, #2, #3)}
\newcommand{\violated}[3]{\textsc{Violated}(#1, #2, #3)}
\newcommand{\act}[3]{\textsc{Active}(#1, #2, #3)}
\newcommand{\obgbyhead}[1]{OBG(#1)}

\newcommand{\events}{E}
\newcommand{\measures}{M}
\newcommand{\ruleset}{\textit{Rules}}
\newcommand{\fos}{\sigma}

\newcommand{\term}{t}
\newcommand{\constant}{c}
\newcommand{\measure}{m}
\newcommand{\event}{e}
\newcommand{\obg}{ob}
\newcommand{\cobg}{cob}
\newcommand{\otherwise}{\textbf{ otherwise }}
\newcommand{\sand}{\textbf{ and }}
\newcommand{\obghead}{h}
\newcommand{\prop}{p}
\newcommand{\aprop}{ap}
\newcommand{\aelement}{ae}
\newcommand{\dobg}{\bigvee_{\cobg}}
\newcommand{\srule}{r}
\newcommand{\fact}{f}
\newcommand{\facts}{Facts}
\newcommand{\eventocc}[1]{\mathcal{E}_{#1}}
\newcommand{\measureassign}[1]{\mathbb{M}_{#1}}
\newcommand{\timestamp}[1]{\delta_{#1}}

\newcommand{\domain}{D}
\newcommand{\variables}{V}

\begin{document}

\copyrightyear{2024}
\acmYear{2024}
\setcopyright{acmlicensed}\acmConference[ICSE '24]{2024 IEEE/ACM 46th International Conference on Software Engineering}{April 14--20, 2024}{Lisbon, Portugal}
\acmBooktitle{2024 IEEE/ACM 46th International Conference on Software Engineering (ICSE '24), April 14--20, 2024, Lisbon, Portugal}
\acmDOI{10.1145/3597503.3639093}
\acmISBN{979-8-4007-0217-4/24/04}
\setlength{\abovedisplayskip}{1pt}
\setlength{\belowdisplayskip}{1pt}

\newcommand{\affilUT}{\affiliation{%
  \institution{University of Toronto}
  \city{Toronto}
  \country{Canada}
}}
\newcommand{\affilUY}{\affiliation{%
  \institution{University of York}
  \city{York}
  \country{UK}
}}
\newcommand{\affilSC}{\affiliation{%
  \institution{Smith College}
  \city{Northampton}
  \country{USA}
}}
\newcommand{\affilUB}{\affiliation{%
  \institution{University of Brasilia}
  \city{Brasilia}
  \country{Brazil}
}}

\author[N. Feng]{Nick Feng}
\email{fengnick@cs.toronto.edu}
\affilUT

\author[L. Marsso et al.]{Lina Marsso}
\email{lina.marsso@utoronto.ca}
\affilUT

\author[ ]{Sinem Getir Yaman}
\email{sinem.getir.yaman@york.ac.uk}
\affilUY

\author[ ]{Beverley Townsend}
\email{bev.townsend@york.ac.uk}
\affilUY

\author[ ]{Yesugen Baatartogtokh}
\email{ybaatartogtokh@smith.edu}
\affilSC

\author[ ]{Reem Ayad}
\email{reem.ayad@mail.utoronto.ca}
\affilUT

\author[ ]{Victoria Oldemburgo de Mello}
\email{victoria.mello@mail.utoronto.ca}
\affilUT

\author[ ]{Isobel Standen}
\email{isobel.standen@york.ac.uk}
\affilUY

\author[ ]{Ioannis Stefanakos}
\email{ioannis.stefanakos@york.ac.uk}
\affilUY

\author[ ]{Calum Imrie}
\email{calum.imrie@york.ac.uk}
\affilUY

\author[ ]{Genaina Rodrigues}
\email{genaina@unb.br}
\affilUB

\author[ ]{Ana Cavalcanti}
\email{ana.cavalcanti@york.ac.uk}
\affilUY

\author[ ]{Radu Calinescu}
\email{radu.calinescu@york.ac.uk}
\affilUY

\author[ ]{Marsha Chechik}
\email{chechik@cs.toronto.edu}
\affilUT



\title[Analyzing and Debugging Normative Requirements via Satisfiability Checking]{
Analyzing and Debugging Normative Requirements \\ via Satisfiability Checking
}
\author{}

\begin{abstract} 
As software systems increasingly interact with humans in application domains such as transportation and healthcare, they raise concerns related to the social, legal, ethical, empathetic, and cultural~(SLEEC) norms and values of their stakeholders. \emph{Normative non-functional requirements} (\nNFRs) are used to capture these concerns by setting SLEEC-relevant boundaries for system behavior. Since \nNFRs{} need to be specified by multiple stakeholders with widely different, non-technical expertise (ethicists, lawyers, regulators, end users, etc.), \nNFR{} elicitation is very challenging. To address this difficult task, we introduce \approach, a novel tool-supported formal approach to \nNFR{} analysis and debugging. \approach employs satisfiability checking to identify a broad spectrum of \nNFR{} well-formedness  issues, such as conflicts, redundancy, restrictiveness, and insufficiency, yielding diagnostics that pinpoint their causes in a user-friendly way that enables non-technical stakeholders to understand and fix them. We show the effectiveness and usability of our approach through nine case studies in which teams of ethicists, lawyers, philosophers, psychologists, safety analysts, and engineers used \approach{} to analyse and debug 233 \nNFRs{}, comprising 62 issues for the software underpinning the operation of systems, such as, assistive-care robots and tree-disease detection drones to manufacturing collaborative robots.
\end{abstract}

\maketitle 

\vspace{-0.1in}
\section{Introduction}
\label{sec:introduction}
Software systems interacting with humans are becoming prevalent in domains such as transportation~\cite{Pham-et-al-16,Nagatani-et-al-21}, 
healthcare~\cite{Salceda-05,Laursen-et-al-22}, and assistive care
~\cite{jevtic2018personalized}. The deployment of such systems raises significant concerns related to the  
\emph{social, legal, ethical, empathetic, and cultural} (SLEEC) impact of their operation, 
e.g., the control software of an assistive-dressing robot~\cite{jevtic2018personalized} can accidentally expose a patient's privacy when dressing them in a room with the curtains open. 
To capture SLEEC concerns effectively, 
researchers~\cite{townsend-et-al-2022, Bremner2019DFW, Burton2020HMP,Anderson_Anderson_2007} have proposed to use \emph{normative non-functional requirements} (\nNFRs) 
that constrain the acceptable range of system behaviors.
For instance, \nNFRs{} for the assistive-dressing robot may specify that the user's well-being and privacy are prioritized by promptly completing the dressing task, especially when the patient is underdressed. 
\nNFRs{} are used to inform all subsequent software engineering (SE) activities, such as design and verification. 

Prior research recommended using \dsl, a state-of-the-art~(SoTA) domain-specific language~\cite{KOSAR201677}, for specifying \nNFRs{} as \dsl{} rules~\cite{Getir-Yaman-et-al-23,Getir-Yaman-et-al-23b}. \dsl{} is very close to natural language and has proven accessible to stakeholders without a technical background (e.g., lawyers, ethicists, and regulators). A \dsl{} rule defines the required system \emph{response} to specific \emph{events}, and can be accompanied by one or more \emph{defeaters}, which  
 specify circumstances that preempt the original response and can optionally provide an alternative response. For instance, a normative rule $r1$ for the assistive-care robot from our earlier example 
can be encoded as ``\textbf{when} \texttt{DressingRequest} \textbf{then} \texttt{DressingStarted} \textbf{within} 10 \textbf{minutes} \textbf{unless} \texttt{curtainsOpen}
''.  Here, \texttt{DressingRequest} is an event, \texttt{DressingStarted} \textbf{within} 10 \textbf{minutes} is a response, and \texttt{curtainsOpen} is a defeater with a condition but no response. 



\nNFRs{} are typically defined by system stakeholders with widely different expertise (e.g., ethicists, lawyers, regulators, and system users). As such, they are susceptible to complex and often subtle \emph{well-formedness issues} (WFI) (i.e., issues that are application- and domain-independent) such as conflicts, redundancies, insufficiencies, overrestrictions, and many others~\cite{Shah-et-al-21, Lamsweerde-Darimont-Letier-98, Matsumoto-et-al-17, Lamsweerde-2009}.  
In particular, the different objectives, concerns, responsibilities, and priorities of these stakeholders 
are sometimes incompatible, leading to logical inconsistency in all circumstances, or are subject to boundary conditions, rendering the requirements infeasible  given their inconsistent prescription of event occurrences~\cite{Lamsweerde-2009}. Therefore, resolving such \CRRI is essential for ensuring the validity of \nNFRs{}. For instance, identifying and removing unnecessary redundant \nNFRs{} simplifies the validation and review process. Additionally, one must ensure that the \nNFRs{} are neither overly permissive in
addressing the relevant SLEEC concerns, nor overly restrictive and preventing the system from achieving its functional requirements. 



 \nNFRs{} present two additional layers of complexity compared to ``regular'' non-functional requirements~(NFRs).
 Firstly, their operationalisation includes challenges due to the involvement of stakeholders who often lack technical expertise.  While this problem is shared with NFRs in general, \nNFR{} stakeholders include experts with very different backgrounds:  lawyers, social workers, ethics experts, etc., and conflicting or redundant requirements can arise from different norms and factors~\cite{townsend-et-al-2022} even if elicited by a single stakeholder!
 Secondly, the task becomes even more complex when capturing intricate non-monotonic conditions expressed via \emph{if, then, unless} and time constraints (e.g., \textit{act within 1 week}), which are often abundant in \nNFRs~\cite{knoks2020,brunero2022,townsend-et-al-2022,feng-et-al-23-b}. The presence of non-monotonic conditions leads to branching choices in reasoning, and the presence of time constraints forces the reasoning to consider events at different time frames.
For example, consider a helper robot with the \nNFR{} of ``protecting patients' privacy'', which might \emph{conflict} with a ``patient claustrophobia'' \nNFR{} if the patient insists on keeping the curtains open while being dressed. 
Additionally, a social \nNFR{} may require respecting patient privacy, while a legal \nNFR{} may specify adherence to data protection laws, potentially creating a redundancy. Resolving these issues may lead to \nNFRs{} that no longer align with the stakeholders goals, either being overly restrictive or insufficient to address SLEEC concerns comprehensively.
Existing work has analyzed certain well-formedness issues, such as redundancy and conflicts, but is limited to interactions between two \dsl rules~\cite{Getir-Yaman-et-al-23}, ignoring specific contexts~\cite{feng-et-al-23-b},
or offering limited debugging assistance that requires some technical expertise, e.g., familiarity with the output format of a  model checker.


As a starting point, in this paper, we propose an approach to the analysis and debugging of \nNFRs{} 
by checking four well-formedness constraints: redundancy, conflicts (accounting for different contexts), restrictiveness, and insufficiency. 
To the best of our knowledge, our solution is the first to provide such an extensive analysis and debugging approach for \nNFRs. Our approach is built on top of satisfiability checking for first-order logic with quantifiers over relational objects (\fol)~\cite{feng-et-al-23}. \fol has been shown effective for modeling requirements that define constraints over unbounded time and data domains~\cite{feng-et-al-23}, and so is well suited for capturing 
metric temporal constraints over actions and (unbounded)
environmental variables, referred to as \emph{measures}, in \dsl.  Moreover, the proof of \fol unsatisfiability enables us to provide diagnosis for the causes of \nNFRs{} WFIs. Capturing  \dsl with \fol  also enables other logic-based analysis such as deductive verification. 
\pdfoutput=1
\tikzset{
    stepop/.style={draw,dotted,rectangle,rounded corners,minimum width=1.5cm,minimum height=0.9cm,align=center,text width=3.3cm,fill=tuftsblue!30,font=\small\sffamily},
    stepopl/.style={draw,dotted,rectangle,rounded corners,minimum width=1.5cm,minimum height=0.9cm,align=center,text width=4cm,fill=tuftsblue!30,font=\small\sffamily},
    rndoutputitt/.style={draw,rectangle,rounded corners,minimum width=2cm,minimum height=1.5cm,align=center,text width=2.4cm,fill=tuftsblue!20,font=\small\sffamily},
    rndoutputi/.style={draw=blue!20,rectangle,rounded corners,minimum width=2cm,minimum height=1.2cm,align=center,text width=2.8cm,fill=blue!20,font=\sffamily},
    rndoutputisl/.style={draw=blue!20,rectangle,rounded corners,minimum width=2cm,minimum height=1.2cm,align=center,text width=3.4cm,fill=blue!20,font=\sffamily},
    rndoutputis/.style={draw,rectangle,rounded corners,minimum width=2cm,minimum height=0.9cm,align=center,text width=3.2cm,fill=tuftsblue!30,font=\small\sffamily},
    rndoutputib/.style={draw,rectangle,rounded corners,minimum width=2cm,minimum height=0.9cm,align=center,text width=3.5cm,fill=tuftsblue!30,font=\small\sffamily},
    rndoutputibse/.style={draw,rectangle,rounded corners,minimum width=2cm,minimum height=0.9cm,align=center,text width=2.2cm,fill=tuftsblue!30,font=\small\sffamily},
    rndoutputibnl/.style={draw,rectangle,rounded corners,minimum width=2cm,minimum height=0.9cm,align=center,text width=4.2cm,fill=tuftsblue!30,font=\small\sffamily},
    rndoutputibp/.style={draw,rectangle,rounded corners,minimum width=2cm,minimum height=0.9cm,align=center,text width=2.45cm,fill=tuftsblue!30,font=\small\sffamily},
    rndoutputibc/.style={draw,line width=1pt,rectangle,rounded corners,minimum width=2cm,minimum height=0.9cm,align=center,text width=4.3cm,dotted,fill=gray!10,font=\small\sffamily},
    rndoutputibcs/.style={draw,line width=1pt,rectangle,rounded corners,minimum width=2cm,minimum height=0.9cm,align=center,text width=2.8cm,dotted,fill=gray!10,font=\small\sffamily},
    rndoutputl/.style={draw,rectangle,rounded corners,minimum width=0.45cm,minimum height=0.4cm,align=center,text width=0.75cm,fill=tuftsblue!30,font=\sffamily},
    rndoutputlg/.style={draw,rectangle,rounded corners,minimum width=0.45cm,minimum height=0.4cm,align=center,text width=0.75cm,line width=1pt,fill=gray!10,font=\sffamily},
    rndoutput/.style={draw,rectangle,rounded corners,minimum width=2cm,minimum height=1.5cm,align=center,text width=2cm,fill=tuftsblue!30,font=\small\sffamily},
    extact/.style={draw,dashed,rectangle,rounded corners,minimum width=2cm,minimum height=1.5cm,align=center,text width=3.4cm,dotted,fill=gray!50,font=\small\sffamily},
    recnorm/.style={draw,minimum width=0.6cm,minimum height=0.65cm,align=center,text width=0.65cm,dotted,fill=lightgray!50,font=\small\sffamily},
    recnormw/.style={draw,minimum width=0.6cm,dashed,minimum height=0.65cm,align=center,text width=0.65cm,fill=white,font=\sffamily},
    reclong/.style= {rounded corners=0.5cm,draw,minimum width=2cm,minimum height=1.5cm,align=center,text width=4cm,dotted,fill=gray!35,font=\sffamily},
    prog/.style={draw,dashed,tape,tape bend top=none,align=center,text width=2cm,fill=turquoisegreen!50,font=\small\sffamily},
    progl/.style={draw,tape,tape bend top=none,align=center,text width=0.7cm,fill=turquoisegreen!50,font=\small\sffamily},
    artifact/.style={draw,trapezium,trapezium left angle=75,trapezium right angle=105,text width=2.2cm,fill=yellow!35,align=center,font=\sffamily},
    artifactm/.style={draw,trapezium,trapezium left angle=83,trapezium right angle=98,text width=2.7cm,fill=cookie!35,align=center,font=\small\sffamily},
    artifacti/.style={draw,trapezium,trapezium left angle=83,trapezium right angle=98,text width=2.3cm,fill=cookie!35,align=center,font=\small\sffamily},
    artifactitts/.style={draw,trapezium,trapezium left angle=83,trapezium right angle=98,text width=1.6cm,fill=cookie!50,align=center,font=\small\sffamily},
    artifactitt/.style={draw,trapezium,trapezium left angle=83,trapezium right angle=98,text width=2cm,fill=cookie!50,align=center,font=\small\sffamily},
    artifactittsss/.style={draw,trapezium,trapezium left angle=83,trapezium right angle=98,text width=1.7cm,fill=cookie!50,align=center,font=\small\sffamily},
    artifactittsssp/.style={draw,trapezium,trapezium left angle=83,trapezium right angle=98,text width=1.35cm,fill=cookie!50,align=center,font=\small\sffamily},
    artifactittlsi/.style={draw,trapezium,trapezium left angle=83,trapezium right angle=98,text width=2.55cm,fill=lightapricot!60,align=center,font=\small\sffamily},
    artifactittls/.style={draw,trapezium,trapezium left angle=83,trapezium right angle=98,text width=2.55cm,fill=cookie!50,align=center,font=\small\sffamily},
    artifactittlms/.style={draw,trapezium,trapezium left angle=83,trapezium right angle=98,text width=2.7cm,fill=cookie!50,align=center,font=\small\sffamily},
    artifactittl/.style={draw,trapezium,trapezium left angle=83,trapezium right angle=98,text width=3cm,fill=cookie!50,align=center,font=\small\sffamily},
    artifactittlsl/.style={draw,trapezium,trapezium left angle=83,trapezium right angle=98,text width=3.2cm,fill=cookie!50,align=center,font=\small\sffamily},
    artifactittll/.style={draw,trapezium,trapezium left angle=83,trapezium right angle=98,text width=3.6cm,fill=cookie!50,align=center,font=\small\sffamily},
    artifactittlll/.style={draw,trapezium,trapezium left angle=83,trapezium right angle=98,text width=3.8cm,fill=cookie!50,align=center,font=\small\sffamily},
    artifactittlls/.style={draw,trapezium,trapezium left angle=83,trapezium right angle=98,text width=3.15cm,fill=cookie!50,align=center,font=\small\sffamily},
    artifacto/.style={draw,trapezium,trapezium left angle=79,trapezium right angle=100,text width=1.5cm,fill=cookie!35,align=center,font=\small\sffamily},
    artifactoi/.style={draw,trapezium,trapezium left angle=79,trapezium right angle=100,text width=1.5cm,fill=lightgray!50,align=center,font=\small\sffamily},
    artifacts/.style={draw,dashed,trapezium,trapezium left angle=79,trapezium right angle=100,text width=3cm,fill=gray!50,align=center,font=\small\sffamily},
    artifactl/.style={draw,trapezium,minimum height=0.35cm,trapezium left angle=75,trapezium right angle=105,text width=0.3cm,fill=white,align=center,font=\sffamily},
    artifactvt/.style={draw,trapezium,trapezium left angle=83,trapezium right angle=98,text width=1cm,fill=cookie!35,align=center,font=\small\sffamily},
    database/.style={draw,cylinder,cylinder uses custom fill,shape border rotate=90,minimum height=1.5cm, minimum width=2.5cm,aspect=0.3},
}

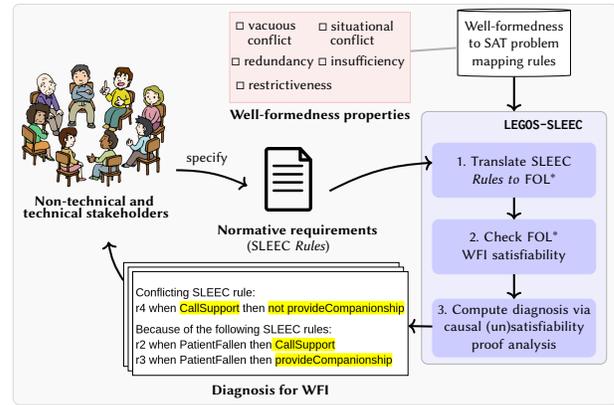
\begin{figure}[t]
  \centering
  \scalebox{0.6}{
  \begin{tikzpicture}
  
  \filldraw[draw=lightgray!60,rounded corners,fill=lightgray!10] (4.6,-3.2) rectangle (-8.8,5.68);
  
  \node[inner sep=0pt] (stakeholders) at (-7,3) {\includegraphics[scale=.15]{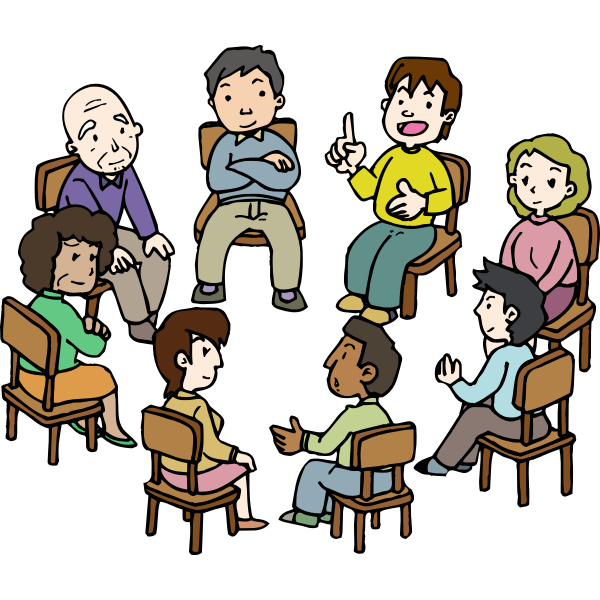}};
  \node[rotate=0,font=\sffamily] at (-7, 1.32) {\textbf{Non-technical and}}; 
  \node[rotate=0,font=\sffamily] at (-7, 1.05) {\textbf{ technical stakeholders}};
  \node[rotate=0,font=\small\sffamily] (stakeholdersLeft) at (-5.3,2) {~}; 
  \node[rotate=0,font=\small\sffamily] (stakeholdersRight) at (-4.9,2.6) {~}; 
  \node[rotate=0,font=\small\sffamily] (stakeholdersBelow) at (-6.8, 0.5) {~}; 

  \node[draw=babypink,fill=babypink!30,rectangle,minimum width=4cm,minimum height=2cm,font=\small\sffamily,line width=0.15mm] (property) at (-2,4.5) {}; 
  \node[rotate=0,font=\sffamily] at (-2,3.2)  {\textbf{\upCRRI properties}};
  \node[rotate=0,font=\sffamily] at (-3.2,5.2)  {$\square$ vacuous};
  \node[rotate=0,font=\sffamily] at (-3.15,4.9)  {~conflict};
  \node[rotate=0,font=\sffamily] at (-1.2,5.2)  {$\square$ situational};
  \node[rotate=0,font=\sffamily] at (-1.32,4.9)  {~conflict};
  \node[rotate=0,font=\sffamily] at (-3.05,4.4)  {$\square$ redundancy};
  \node[rotate=0,font=\sffamily] at (-1.08,4.4)  {$\square$ insufficiency};
  \node[rotate=0,font=\sffamily] at (-2.8,3.9)  {$\square$ restrictiveness};

  \node[inner sep=0pt] (russell) at (-2.7,1.8) {\includegraphics[scale=.1]{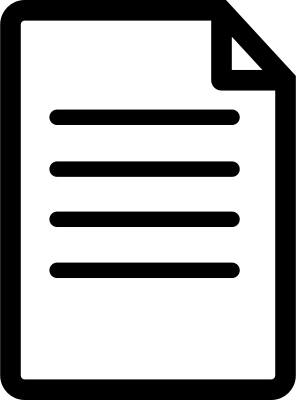}};
  \node[rotate=0,font=\sffamily] at (-2.5,0.65) {\textbf{Normative requirements}};
  \node[rotate=0,font=\sffamily] at (-2.7,0.3) {(SLEEC \textit{Rules})};
  \node[rotate=0,font=\small\sffamily] (rulesLeft) at (-3.5,1.6) {~}; 
  \node[rotate=0,font=\small\sffamily] (rulesRight) at (-1.9,1.5) {~}; 
  
    \node[draw,fill=white,rectangle,minimum width=6.1cm,minimum height=2.3cm,font=\small\sffamily,line width=0.15mm] (diagn) at (-3.3,-1.2) {};
    \node[draw,fill=white,rectangle,minimum width=6.1cm,minimum height=2.3cm,font=\small\sffamily,line width=0.15mm] (diag2) at (-3.2,-1.3) {};
    \node[draw,fill=white,rectangle,minimum width=6.1cm,minimum height=2.3cm,font=\small\sffamily,line width=0.15mm] (diag) at (-3.1,-1.4) {};
   \node[inner sep=0pt] (diag1) at (-3.1,-1.5) {\includegraphics[scale=.23]{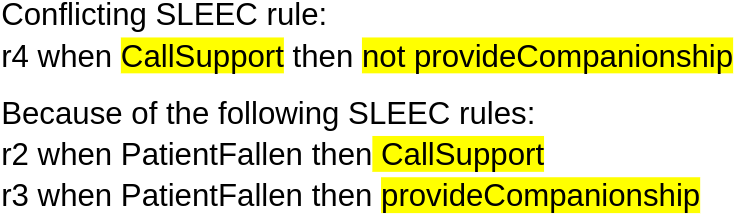}};
  \node[rotate=0,font=\sffamily] at (-3.1,-2.9) {\textbf{Diagnosis for \CRRI}};

  \node[database] (satisfiability) at (2.3,4.8) {~};
  \node[rotate=0,font=\sffamily] at (2.3,5.2) {\upCRRI};
  \node[rotate=0,font=\sffamily] at (2.3,4.8) {to SAT problem};
  \node[rotate=0,font=\sffamily] at (2.3,4.4) {mapping rules};

  \node[draw=gray!60,rectangle,fill=lavender!70,rounded corners,minimum width=4.1cm,minimum height=5.6cm,font=\small\sffamily,line width=0.15mm] (approach) at (2.3,0.5) {};
  \node[rotate=0,font=\sffamily] at (2.95,3) {\textbf{\tool}};
  \node[rndoutputisl] (translate) at (2.3,2) {1. Translate SLEEC \textit{Rules} $to$ \fol};
  \node[rndoutputisl] (check) at (2.3,0.3) {2. Check \fol \CRRI satisfiability};
  \node[rndoutputisl] (buildDiag) at (2.3,-1.5) {3. Compute diagnosis via causal (un)satisfiability proof analysis};
  
  
  \draw[->,line width=0.45mm] (stakeholdersLeft) to  [bend left=12](rulesLeft);
  \node[rotate=0,font=\sffamily] at (-4.5,2.25) {specify};
  \draw[->,line width=0.45mm] (diagn) to  [bend left=12] (stakeholdersBelow);

  \draw[-,line width=0.45mm,gray!60]  (property) to (satisfiability);
  \draw[->,line width=0.45mm]  (satisfiability) to (approach);
  \draw[->,line width=0.45mm] (translate) to (check);
  \draw[->,line width=0.45mm] (check) to (buildDiag);
  \draw[->,line width=0.45mm] (rulesRight) to [bend left=12](translate);
  \draw[->,line width=0.45mm] (buildDiag) to (diag);
  \end{tikzpicture}
  }
  \vspace{-0.1in}
\caption{{\small Overview of the \approach{} approach.}}
\label{fig:overview}	
\vspace{-0.10in}
\end{figure}

\vspace{-0.05in}
\boldparagraph{Contributions}
We contribute
 (a) a method for identifying well-formedness issues (\CRRI)  in \nNFRs{} expressed in \dsl by converting them into satisfiability checking problems;
(b) a method for generating a non-technical diagnosis that pinpoints the causes of \CRRI  by analyzing the causal proof of (un)satisfiability;
 (c) a validation of the effectiveness of the approach
 over nine cases studies and 233 rules written by eight non-technical and technical stakeholders.  
 Our approach is implemented in a tool \tool -- see Fig.~\ref{fig:overview}.   \tool 
translates \nNFRs{} expressed as rules in the SoTA \dsl~\cite{Getir-Yaman-et-al-23,Getir-Yaman-et-al-23b} 
into \fol;
checks their satisfiability using the SoTA \fol{} checker; and generates a non-technical diagnosis.  
 
%

\vspace{-0.05in}
\boldparagraph{Significance}
A comprehensive approach for analysis and debugging of \nNFRs{} is crucial in developing modern systems that increasingly interact with humans.  To the best of our knowledge, we are  the first to provide such an approach, supporting stakeholders with varying levels of technical expertise. The diagnosis computed by \approach can be used by software engineers to \mbox{(semi-)automatically} generate and suggest patches for resolving \CRRI, thereby further assisting the technical and non-technical stakeholders.

The rest of this paper is organized as follows.
Sec.~\ref{sec:background} gives the background material for our work.
Sec.~\ref{sec:problem} defines the \nNFR{} \CRRI  problem that we tackle in this paper.
Sec.~\ref{sec:sanitization} describes our approach, \approach, to identify \nNFRs{} \CRRI  using \fol satisfiability checking.
Sec.~\ref{sec:diagnose} presents our approach to utilizing the causal proof of (un)satisfiability to highlight the reasons behind  \CRRI  (e.g., conflict).
Sec.~\ref{sec:evaluation} presents the evaluation of the approach's effectiveness.
Sec.~\ref{sec:relatedwork} compares \approach to related work.
Sec.~\ref{sec:conclusion} discusses future research.







\section{Preliminaries}
\label{sec:background}

\begin{table}[t]
    \centering
    \caption{\small Syntax of normalized $\dsl$ with signature ($\events$, $\measures$).}
    \vspace{-0.1in}
    \scalebox{0.8}{
        \begin{tabular}{ll}
        \toprule
        Name & Definition\\
        \hline
        Term & $\term$ := $\constant :\mathbb{N} \mid \measure\in \measures \mid -\term \mid \term + \term \mid \constant \times \term$  \\
        Proposition & $\prop$ := $\top \mid \bot \mid \term = \term \mid \term \ge \term \mid \notS \; \prop \mid \prop \; \andS \; \prop \mid \prop \; \orS \; \prop$ \\
        Obligation & $\obg^+$ := $ \within{\event}{\term}$ $ \mid$ $\obg^-$ := $\within{\notS \; \event}{\term}$ \\
        Cond Obligation & $\cobg^+$ := $\prop \Rightarrow \obg^+$ $\mid$ $\cobg^-$ := $\prop \Rightarrow \obg^-$ \\
        Obligation Chain & $\dobg $ := $\cobg \mid \cobg^+ \otherwise \dobg$ \\
        Rule & $\srule $ := $\rulesyntax{\event \sand \prop }{\dobg}$ \\
        \hline
        Fact & $\fact$ := $\factsyntax{\event \sand \prop}{(\dobg \mid \notS \dobg)}$\\
        \bottomrule
        \end{tabular}
    }
 \label{tab:syntax}
 \vspace{-0.10in}
\end{table}

In this section, we overview the SoTA \dsl for specification of normative requirements and the \fol satisfiability checker. 

\subsection{Normalized Representation of \dsl}\label{ssec:ndsl}
 \dsl is an event-based language~\cite{Getir-Yaman-et-al-23} for specifying normative requirements with the pattern ``\textbf{when} trigger \textbf{then} response''. The \dsl is based on propositional logic enriched with temporal constraints (e.g., \textbf{within} x \textbf{minutes}) and the constructs \textbf{unless}, specifying \emph{defeaters}, and \textbf{otherwise}, specifying \emph{fallbacks}.

\dsl is an expressive language, enabling specification of complex nested defeaters and responses.
A \dsl rule can have many logically equivalent expressions, posing a challenge for analysis and logical reasoning.  To handle it, the \textit{normalized} \dsl introduced here can be used. In its rules, nested response structures are flattened; normalisation eliminates defeaters and captures their semantics using propositional logic.
%
Here, we present the syntax and semantics of the \textit{normalized} \dsl\footnote{A translation from \dsl to the normalized form is provided in Appendix~\ref{ap:norm} and implemented in \tool.} and then define the accepted behaviour specified by normalized rules.



\vspace{-0.05in}
\noindent
\boldparagraph{Normalized \dsl syntax}
A normalised \dsl \emph{signature} is a tuple $\sig = (\events, \measures)$, where $\events = \{\event_1, \ldots, \event_n\}$ and $\measures = \{\measure_1 \ldots \measure_n\}$ are finite sets of symbols  for \textit{events}  and \textit{measures}, respectively. As a convention, in the rest of the paper, we assume that all event symbols (e.g., {\small \texttt{OpenCurtain}}) are capitalized while measure symbols (e.g., {\small \texttt{underDressed}}) are not. 
Without loss of generality, we assume that every measure is numerical. The syntax of the normalized \dsl{} can be found in Tbl.~\ref{tab:syntax}, where, e.g., an \textit{obligation} $\obg = \within{\obghead}{\term}$ specifies that the event $\event$ must (if $\obghead = \event$) or must not  (if $\obghead = \notS \; \event$) occur within the time limit $\term$. We omit the time limit $t$ if $t = 0$.

\vspace{-0.05in}
\begin{example}\label{example:syntax}
    Consider the normative rule ($r5$) ``\textbf{when} {\small\texttt{OpenCurtain\-Request}} \textbf{then} {\small\texttt{OpenCurtain}} \textbf{within} 30 {\small\textbf{minutes}} \textbf{unless} {\small \texttt{underDressed}}'', where {\small \texttt{OpenCurtainRequest}} and {\small \texttt{OpenCurtain}} are events and {\small \texttt{under\-Dressed}} is a measure. $r5$ is normalized as ``\textbf{when} {\small \texttt{OpenCurtainRequest}} \textbf{and not} {\small \texttt{underDressed}} \textbf{then} {\small \texttt{OpenCurtain}} \textbf{within} 30 \textbf{minutes}''.
\end{example}

\vspace{-0.1in}
\noindent
\boldparagraph{Semantics}
The semantics of \dsl is described over  \textit{traces}, which are \textit{finite} sequences of \textit{states} $\fos = (\eventocc{1}, \measureassign{1}, \timestamp{1}), (\eventocc{2}, \measureassign{2}, \timestamp{2}),$ $\ldots  (\eventocc{n}, \measureassign{n}, \timestamp{n})$. For every time point $i \in [1, n]$, (1) $\eventocc{i}$ is a set of events that occur at time point $i$; (2) $\measureassign{i}: \measures \rightarrow \mathbb{N}$ assigns every measure in $\measures$ to a concrete value at time point $i$, and (3) $\timestamp{i}: \mathbb{N}$ captures the time value of time point $i$ (e.g., the second time point can have the value 30, for 30 sec). We assume that the time values in the trace are strictly increasing (i.e., $\timestamp{i} < \timestamp{i+1}$ for every $i \in [1, n-1]$). Given a measure assignment $\measureassign{i}$ and a term $\term$, let $\measureassign{i}(\term)$ denote the result of substituting every measure symbol $\measure$ in $\term$ with $\measureassign{i}(\measure)$. Since $\measureassign{i}(\term)$ does not contain free variables, the substitution results in a natural number. 
 Similarly, given a proposition $\prop$, we say that $\measureassign{i} \models \prop$ if $\prop$ is evaluated to $\top$ after substituting every term $\term$ with $\measureassign{i}(\term)$.  
Let a trace $\fos = (\eventocc{1}, \measureassign{1}, \timestamp{1}) \ldots (\eventocc{n}, \measureassign{n}, \timestamp{n})$ be given.

\italicparagraph{Positive obligation} A positive obligation $\within{\event}{\term}$ is \textit{fulfilled subject to time point $i$}, denoted $\fos \models_{i} \within{\event}{\term}$, if there is a time point $j \ge i$ such that $e \in \eventocc{j}$ and $\timestamp{j} \in [\timestamp{i}, \timestamp{i} + \measureassign{i}(\term)]$. That obligation is \emph{violated at time point $j$}, denoted as $\fos \not\models_{i}^{j} \within{\event}{\term}$, if $\timestamp{j} = \timestamp{i} + \measureassign{i}(\term)$ and for every $j'$ such that $j \ge j' \ge i$, $\event$ does not occur ($\event \not\in \eventocc{j'}$).  

\italicparagraph{Negative obligation} A negative obligation $\within{\notS \; \event}{\term}$ is fulfilled subject to time point $i$, denoted as  $\fos \models_{i} \within{\notS \; \event}{\term}$, if for every time point $j$ such that $\timestamp{j} \in [\timestamp{i}, \timestamp{i} + \measureassign{i}(\term)]$, $e \not\in \eventocc{j}$. 
The negative obligation is violated at time point $j$, denoted as $\fos \not\models_{i}^{j} \within{\notS \; \event}{\term}$, if (1) $\timestamp{j} \in [\timestamp{i}, \;\timestamp{i} + \measureassign{i}(\term)]$, (2) $\event$ occurs ($\event \in \eventocc{j}$), and (3) for every $j \ge j' \ge i$, $\event$ does not occur ($\event \not\in \eventocc{j'}$). 

\italicparagraph{Conditional obligation} A conditional obligation $\prop \Rightarrow \obg$ is fulfilled subject to time point $i$, denoted as $\fos \models_{i} (\prop \Rightarrow \obg)$, if $\prop$ does not hold at time point $i$ ($\measureassign{i}(\prop) = \bot$) or the obligation is fulfilled ($\fos \models_{i} \obg$).  Moreover, $\prop \Rightarrow \obg$  is violated at time point $j$, denoted $\fos \not\models_{i}^{j} (\prop \Rightarrow \obg)$, if  $\prop$ holds at $i$ ($\measureassign{i}(\prop) = \top$) and $\obg$ is violated at $j$ ($\fos \not\models_{i}^{j} \obg$).

\vspace{-0.02in}
\italicparagraph{Obligation chain} A chain $\cobg^+_1 \otherwise \cobg^+_2 \ldots \cobg_m$ is fulfilled subject to time point $i$, denoted as $\fos \models_{i} \cobg^+_1 \otherwise$ $\cobg^+_2 \ldots \cobg_m$, if (1) the first obligation is fulfilled ($\fos \models_{i} \cobg^+_1$) or (2) 
there exists a time point $j \ge i$ such that $\cobg^+_1$ is \textit{violated} (i.e., $\fos \not\models_{i}^{j} \cobg_1$) and the rest of the obligation chain (if not empty) is fulfilled at time point $j$ ($\fos \models_{j} \cobg^+_2 \otherwise \ldots \cobg_m$).

\vspace{-0.02in}
\italicparagraph{Rule} A rule $\rulesyntax{\event \; \andS \; \prop}{\dobg}$ is fulfilled in $\fos$, denoted as $\fos \models \rulesyntax{\event \; \andS \; \prop}{\dobg}$, if for every time point $i$, 
where event $\event$ occurs ($\event \in \eventocc{i}$) and $\prop$ holds ($\measureassign{i}(\prop) = \top$), the obligation chain is fulfilled subject to time point $i$ ($\fos \models_{i} \dobg$). A \emph{trace fulfills a rule set} $\ruleset$, denoted as $\fos \models \ruleset$, if it fulfills every rule in the set (i.e., for every $\srule \in \ruleset$, we have $\fos \models \srule$). 

\vspace{-0.05in}
\begin{example}
    Consider the normalized \dsl rule $r5$ in Ex.~\ref{example:syntax} and the trace $\sigma_1$ shown in Fig.~\ref{fig:traces}, corresponding to $(\eventocc{1}, \measureassign{1}, \timestamp{1})$, where $\eventocc{1} =$ {\small \texttt{OpenCurtainRequest}}, $\measureassign{1}({\small \texttt{underDressed}}) = \top$, and $\timestamp{1} = 1$. The trace $\sigma_1$ fulfills the rule $r5$  (i.e., $\sigma_1 \models r5$), because the rule's triggering condition (i.e., $\notS \; {\small \texttt{underDressed}}$) is not satisfied  (i.e., the user is under-dressed at the time of the request). 
\end{example}

\vspace{-0.11in}
\begin{definition}[Behaviour defined by $\ruleset$]\label{def:behav}
    Let a rule set  $\ruleset$ be given. The accepted \emph{behaviour} defined by $\ruleset$, denoted as $\lang{\ruleset}$, is the largest set of traces such that every trace $\fos$ in $\lang{\ruleset}$ respects every rule $r$ in $\ruleset$, i.e., $\fos \in \lang{\ruleset}$, $\fos \models r$.
\end{definition}
\begin{figure*}
    \centering
    \begin{tabular}{c c}
        \begin{tabular}{c c c }
            \scalebox{.75}{
                \begin{tikzpicture}[scale=2]
                  \node at (1.9,0.9) {\small\color{mygray}time};
                  
                  \node at (0.5,1.24) {\small\color{mymauve}OpenCurtainRequest};
                  \node at (0.5,1.44) {\small\color{mymauve}underDressed};
                  \node at (1.5,1.24) {\small\color{mymauve}$\emptyset$};
                  
                  \node at (0.4,0.92) {\small\color{mygray}$1$};
                  \node at (1.5,0.92) {\small\color{mygray}$30$};
                  
                  \node at (-0.5,1.0) {$\sigma_1$};
            
                  \draw [thick,black,->] (-0.3,1) -- (2.2,1);
                  \draw [thick,black,-] (-0.3,1.05) -- (-0.3,0.95);
                  \draw [thick,black,->] (0.4,1) -- (0.4,1.15);
                  \draw [thick,black,->] (1.5,1) -- (1.5,1.15);
                  
                \end{tikzpicture}
            }
        &  
        \scalebox{.75}{
            \begin{tikzpicture}[scale=2]
              \node at (2.37,0.9) {\small\color{mygray}time};
              
              \node at (0.5,1.24) {\small\color{mymauve}OpenCurtainRequest};
              \node at (1.8,1.24) {\small\color{mymauve}OpenCurtain};
              \node at (1.8,1.44) {\small\color{mymauve}underDressed};
              
              \node at (0.4,0.92) {\small\color{mygray}$1$};
              \node at (1.8,0.92) {\small\color{mygray}$3$};
              
              \node at (-0.5,1.0) {$\sigma_2$};
        
              \draw [thick,black,->] (-0.3,1) -- (2.6,1);
              \draw [thick,black,-] (-0.3,1.05) -- (-0.3,0.95);
              \draw [thick,black,->] (0.4,1) -- (0.4,1.15);
              \draw [thick,black,->] (1.8,1) -- (1.8,1.15);
              
            \end{tikzpicture}
        }
        &
        \scalebox{.75}{
            \begin{tikzpicture}[scale=2]
              \node at (2.1,0.9) {\small\color{mygray}time};
              
              \node at (0.5,1.24) {\small\color{mymauve}OpenCurtainRequest};
              \node at (1.6,1.24) {\small $\{ \emptyset \}$};
              
              \node at (0.4,0.92) {\small\color{mygray}$1$};
              \node at (1.6,0.92) {\small\color{mygray}$30$};
              
              \node at (-0.5,1.0) {$\sigma_3$};
        
              \draw [thick,black,->] (-0.3,1) -- (2.4,1);
              \draw [thick,black,-] (-0.3,1.05) -- (-0.3,0.95);
              \draw [thick,black,->] (0.4,1) -- (0.4,1.15);
              \draw [thick,black,->] (1.6,1) -- (1.6,1.15);
              
            \end{tikzpicture}
        }
    \end{tabular}
             
    \\
    
    \scalebox{.8}{
        \begin{tikzpicture}[scale=2]
          \node at (2.4,0.9) {\small\color{mygray}time};
          
          \node at (0.3,1.24) {\small\color{mymauve}SupportCalled};
          \node at (1.6,1.24) {\small\color{mymauve}OpenCurtainRequest};
          
          \node at (0.25,0.92) {\small\color{mygray}$1$};
          \node at (1.6,0.92) {\small\color{mygray}$11$};
          
          \node at (-0.5,1.0) {$\fos_0^2$};
    
          \draw [thick,black,->] (-0.3,1) -- (3,1);
          \draw [thick,black,-] (-0.3,1.05) -- (-0.3,0.95);
          \draw [thick,black,->] (0.26,1) -- (0.26,1.15);
          \draw [thick,black,->] (1.6,1) -- (1.6,1.15);
          
        \end{tikzpicture}
    }
       \scalebox{.8}{
            \begin{tikzpicture}[scale=2]
              \node at (3.4,1.1) {\small\color{mymauve}};
              \node at (2.7,0.9) {\small\color{mygray}time};
              
              \node at (0.5,1.24) {\small\color{mymauve}OpenCurtainRequest};
              \node at (2,1.44) {\small\color{mymauve}underDressed};
              \node at (2,1.24) {\small\color{mymauve}OpenCurtain};
              
              \node at (0.4,0.92) {\small\color{mygray}$1$};
              \node at (2,0.92) {\small\color{mygray}$3$};
              
              \node at (-0.5,1.0) {$\sigma_5$};
        
              \draw [thick,black,->] (-0.3,1) -- (3,1);
              \draw [thick,black,-] (-0.3,1.05) -- (-0.3,0.95);
              \draw [thick,black,->] (0.4,1) -- (0.4,1.15);
              \draw [thick,black,->] (2,1) -- (2,1.15);
              
            \end{tikzpicture}
        }
    \end{tabular}
    \vspace{-0.15in}
    \caption{{\small Five traces from the dressing robot example. 
    $\sigma_0^2$ is a partial trace, from 0 to 2, used to illustrate  situational-conflicts.}}
    \label{fig:traces}
    \vspace{-0.1in}
\end{figure*}
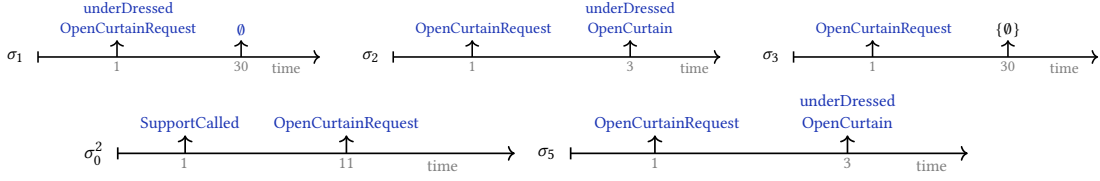

\vspace{-0.1in}
\subsection{\fol Satisfiability Checking}
Here, we present the background on first-order logic with relational objects (\fol)~\cite{feng-et-al-23}, which we use for satisfiability checking.

We start by introducing the syntax of \fol. 
A \emph{signature} $S$ is a tuple $(C, R, \iota)$, where $C$ is a set of constants, $R$ is a set of class symbols, and $\iota : R \rightarrow \mathbb{N} $ is a function that maps a relation to its arity.   We assume that $C$'s domain is $\mathbb{Z}$, where the theory of linear integer arithmetic (LIA) holds. Let $\variables$ be a set of variables in the domain $\mathbb{Z}$. A \emph{relational object} $o$ of class $r \in R$ (denoted as $o:r$) is an object with $\iota(r)$ regular attributes and one special attribute, where every attribute is a variable. We assume that all regular attributes are ordered and denote $o[i]$ to be the $i$th attribute of $o$. 
Attributes are also named, and $o.x$ refers to $o$'s attribute with the name `$x$'.  
Each relational object $o$ has a special attribute, $o.ext$. This attribute is a boolean variable indicating whether $o$ exists in a solution. 
An \fol \emph{term} $t$ is defined inductively as $t: \; c \; | \;  v \; | \; o[k] \; | \; o.x \; | \; t + t \; | \;  c \times t$ for any constant $c \in C$, any variable $v \in \variables$, any relational object $o:r$, any index $k \in [1, \iota(r)]$ and any valid attribute name $x$. Given a signature $S$, \fol formulas take the following forms:
    \emph{(1)} $\top$ and $\bot$, representing ``true'' and ``false''; 
    \emph{(2)} $t = t'$ and $t > t'$, for terms $t$ and $t'$;  
    \emph{(3)} $\phi_f \wedge \psi_f$, $\neg \phi_f$ for \fol formulas $\phi_f$ and ${\psi_f}$; 
    \emph{(4)} $\exists o:r \cdot \: (\phi_f)$ for a \fol formula $\phi_f$  and a class $r$;
    \emph{(5)} $\forall o:r \cdot \: (\phi_f)$ for a \fol formula $\phi_f$  and a class $r$. The quantifiers for \fol formulas are over relational objects of a class, as defined in (4) and (5).  
    Operators $\vee$ and $\forall$ are defined as $\phi_f \vee \psi_f = \neg (\neg \phi_f \wedge \neg \psi_f)$ and $\forall o:r \cdot \phi_f = \exists o:r \cdot \neg \phi_f$. We say that a \fol formula is in \textit{negation normal form} (NNF) if negations ($\neg$) do not appear in front of $\neg$, $\wedge$, $\vee$, $\exists$ and $\forall$. For the rest of the paper, we assume that every formula is in NNF.     
    
Given a signature $S$, \emph{a domain} $\domain$ is a finite set of relational objects. A \fol formula \emph{grounded} in the domain $\domain$ (denoted by $\phi_{\domain}$) is a quantifier-free FOL formula that eliminates quantifiers on relational objects using the following expansion rules: (1) $\exists o: r \cdot \: (\phi_f)$  to $ \bigvee_{{o':r} \in \domain} (o'.ext \wedge \phi_f [o \leftarrow o'])$ and (2) $\forall o:r \cdot \: (\phi_f)$ to $\bigwedge_{{o':r} \in \domain} (o'.ext \Rightarrow \phi_f [o \leftarrow o'])$. A \fol formula $\phi_f$ is \emph{satisfiable in $\domain$} if there exists a variable assignment $v$  that evaluates $\phi_{\domain}$ to $\top$ according to the standard semantics of FOL. A \fol formula $\phi_f$ is \emph{satisfiable} if there exists a finite domain $\domain$ such that $\phi_f$ is satisfiable in $\domain$. 
We call $\sigma = (\domain, v)$ \emph{a satisfying solution to $\phi_f$}, denoted as $\sigma \models \phi_f$. Given a solution $\sigma = (D, v)$, we say a relational object $o$ is in $\sigma$, denoted as $o \in \sigma$, if $o \in D$ and $v(o.ext)$ is true. Note that the solution $\sigma$ contains a subset of relational objects in the domain $\domain$. 
 The \emph{volume of the solution}, denoted as $vol(\sigma)$, is $|\{o \mid o \in \sigma\}|$.

\vspace{-0.06in}
\begin{example} Let $A$ be a class of relational objects with attribute $val$. The formula $\forall a: A.\: (\exists a':A \cdot \: (a.val < a'.val)) \wedge \: \exists a'':A \cdot (a''.val = 0)) $  has no satisfying solutions in any finite domain. On the other hand, the formula $\forall a:A \cdot \: (\exists a', a'':A  \cdot \: (a.val = a'.val + a''.val) \wedge \: (a.val = 5)) $  has a solution $\sigma = (\domain, v)$ of volume 2, with the domain $\domain = (a_1, a_2)$ and the value function $v(a_1.val) = 5$, $v(a_2.val) = 0$ because, if $a \gets a_1$, then the formula is satisfied by assigning $a' \gets a_1, \;a'' \gets a_2$; and, if $a \gets a_2$, then the formula is satisfied by assigning $a' \gets a_2, \;a'' \gets a_2$. 
\end{example}


 \begin{table}
    \caption{{\small Examples of \dsl rules and facts.}}
    \vspace{-0.1in}
    \scalebox{0.78}{
        \begin{tabular}{ c l }
        \toprule
        \makecell{ID} & \makecell{\dsl Rule or Fact}\\
        \toprule
        \multirow{ 2}{*}{$r5$} &
         \textbf{when} \texttt{OpenCurtainRequest} \textbf{and not} \texttt{underDressed}  \\
         & ~~~~\textbf{then} \texttt{OpenCurtain} \textbf{within} 30 \textbf{min.} \\
         \multirow{ 2}{*}{$r5'$} &
         \textbf{when} \texttt{OpenCurtainRequest} \textbf{and not} \texttt{underDressed}  \\
         & ~~~~\textbf{then} \textbf{not} \texttt{OpenCurtain} \textbf{within} 30 \textbf{min.} \\
        $r6$ & \textbf{when} \texttt{OpenCurtainRequest} \textbf{then} \texttt{SignalOpenCurtain} \textbf{unless} \texttt{underDressed} \\
        $r7$ & \textbf{when} \texttt{SignalOpenCurtain} \textbf{then} \texttt{OpenCurtain} \textbf{within} 20 \textbf{min.} \\
        $r8$ & \textbf{when} \texttt{OpenCurtainRequest} \textbf{then} \textbf{not} \texttt{OpenCurtain} \textbf{within} 40 \textbf{min.} \\
            $r9$ & \textbf{when} \texttt{UserFallen} \textbf{then} \texttt{SupportCalled} \textbf{within} 10 \textbf{min.} \\
            $r10$ & \textbf{when} \texttt{SupportCalled} \textbf{then} \texttt{LeaveUser} \textbf{within} 15 \textbf{min.} \\
            $r11$ & \textbf{when} \texttt{SupportCalled} \textbf{then} \texttt{OpenCurtain} \textbf{within} 40 \textbf{min.} \\
            \multirow{ 2}{*}{$r12$} &
         \textbf{when} \texttt{OpenCurtainRequest} \textbf{and not} \texttt{underDressed} \\
         & \textbf{then} \texttt{OpenCurtain} \textbf{within} 20 \textbf{min.}\\
         $r13$ & \textbf{when} \texttt{A} \textbf{then} \texttt{B} \textbf{within} 10 \textbf{min.} \\
         $r14$ & \textbf{when} \texttt{A} $\andS{}$ $P(A)$ \textbf{then} \texttt{B} \textbf{within} 20 \textbf{min.}
          \\
            $\purposeFact_1$ & \textbf{exists} \texttt{UserFallen} \textbf{then} \textbf{not} \texttt{LeaveUser} \textbf{within} 30 \textbf{min.} \\
            $\fact_{1}$ & $\textbf{exists} \; \texttt{OpenCurtainRequest} \; \andS \; \notS \; \texttt{underDressed}$ \\
            \bottomrule
        \end{tabular}
    }
    \label{tab:examplefig}
    \vspace{-0.1in}
    \end{table}

\section{Well-Formedness Properties}\label{sec:problem}
In this section, we present a set of \nNFR{} well-formedness properties:
conflicts, redundancies, restrictiveness, and insufficiency.

\vspace{-0.03in}
\boldparagraph{Redundancy}
A rule 
in a rule set 
is \emph{redundant} if 
removing it does increase the behaviours allowed by the rule set.  
\begin{definition}[Redundancy]\label{def:redudent}
Given a rule set $\ruleset$,  a rule $\srule \in \ruleset$, 
is \emph{redundant}, 
denoted as $(\ruleset \setminus \srule) \Rightarrow \srule$, if $\lang{\ruleset \setminus \srule} \subseteq \lang{\ruleset}$. 
\end{definition}
\begin{example}\label{exmaple:redudent}
    Consider a rule set $\ruleset = \{r5, r6, r7\}$, where $r5$, $r6$ and $r7$ are rules in Tbl.~\ref{tab:examplefig}.
    The rule $r5$ is redundant in $\ruleset$ since the behaviors allowed by $\lang{r6} \bigcap \lang{r7}$ are \textit{subsumed} by the behaviors allowed by $\lang{r5}$ (i.e., $\lang{r6} \bigcap \lang{7} \subset\lang{r5}$), and thus removing $r5$ does not affect the behaviors allowed by $\lang{\ruleset}$. If the time limit of $r7$ is changed to $40$, then $r5$ is no longer redundant. 
\end{example}

\vspace{-0.08in}
\noindent
\boldparagraph{Vacuous Conflicts}
A rule is \emph{vacuously conflicting} if the trigger of the rule is not in the accepted behaviours defined by the rule set, i.e., 
triggering the rule in any situation will lead to a violation of some rules and thus cause a conflict. 

\vspace{-0.07in}
\begin{definition}[Vacuous Conflict]\label{def:dconflict}
Let a rule set $\ruleset$ be given. A rule $\srule  = \rulesyntax{\event \; \andS \; \prop}{\dobg}$ in $\ruleset$  is \textit{vacuously conflicting} if for every  trace $(\eventocc{1}, \measureassign{1}, \timestamp{1}) \ldots (\eventocc{n}, \measureassign{n}, \timestamp{n}) \in \lang{\ruleset}$, $\event \not\in \eventocc{i}$ or $\measureassign{\prop} = \bot$ for every $i \in [1, n]$.
\end{definition}
\begin{example}\label{example:dconf}
    The rule $r5$ (Tbl.~\ref{tab:examplefig}) is conflicting in the rule set $\ruleset = \{r5, r8\}$ since triggering $r5$ would also trigger $r8$.  Thus {\small\texttt{OpenCurtain}} must occur within 30 minutes ($r5$) but must not occur within 40 minutes ($r8$), which is a conflict.
\end{example}

\vspace{-0.08in}
\boldparagraph{Restrictive or Insufficient Rules}
Recall from Def.~\ref{def:behav}, given a set of rules $\ruleset$, $\lang{\ruleset}$ is the set of behaviors permitted by $\ruleset$. If $\ruleset$ is overly restrictive, then $\lang{\ruleset}$ might not contain some desirable behaviors essential to achieve the initial system purpose. 
On the other hand, if $\ruleset$ is too relaxed, then $\lang{\ruleset}$ might contain some undesirable behaviors, such as the ones causing SLEEC concerns.  
To capture the desirable and undesirable behaviors, we introduce a new \dsl language construct \textit{fact}.

\begin{definition}[Fact]\label{def:fact}
    A \dsl \emph{fact} $\fact$ has the syntax
   ``$\factsyntax{\event \; \andS \;$\\$ \prop} {\dobg}$'' or ``$\factsyntax{\event \; \andS \; \prop}{\notS \; \dobg}$'', where $\event$ is an event, $\prop$ is a proposition and $\dobg$ is an obligation chain (see Tbl.~\ref{tab:syntax}). 
    The fact  $\factsyntax{\event \; \andS \; \prop}{\dobg}$ is \emph{satisfied by a trace} $\fos = (\eventocc{1}, \measureassign{1}, \timestamp{1}),$ $ (\eventocc{2}, \measureassign{2}, \timestamp{2}),$ $\ldots  (\eventocc{n}, \measureassign{n}, \timestamp{n})$,  denoted as $\fos \models \fact$, if there exists a time point $i$ such that (1) $\event \in \eventocc{i}$, (2) $\measureassign{i}(\prop) = \top$, and the response $\dobg$ is fulfilled at time point $i$ (i.e., $\fos \models_{i} \dobg$). Similarly, the fact  $\factsyntax{\event \; \andS \; \prop}{\notS \; \dobg}$ is \emph{satisfied by a trace $\fos$ if there exists a time point $i$ such that the conditions (1)-(2) hold, and $\dobg$ is \emph{not} fulfilled at time point $i$}.
     A set of facts $\facts$ \emph{is satisfied by a trace} $\fos$, denoted as $\fos \models \facts$, if  $\fos \models \fact$ for every $\fact \in \facts$. 
    The behavior of $\facts$, denoted as $\lang{\facts}$, is the largest set of traces such that for every  $\fos \in \lang{\facts}$, $\fos \models \facts$.
\end{definition}
Similarly to rules, facts assert invariants over traces. However, unlike rules, which define responses to their triggering events that may or may not happen, facts require occurrences of events.
\begin{definition}[Restrictiveness]\label{def:restrictivness}
    Given a rule set $\ruleset$ and a desirable behavior captured by a set of facts $\purposeFact$, $\ruleset$ is \emph{overly-restrictive for $\purposeFact$} if and only if $\lang{\ruleset} \cap \lang{\purposeFact} = \emptyset$.
\end{definition}
\noindent
Intuitively, a rule set is overly restrictive if it is impossible to execute the desirable behavior while respecting every rule.

\begin{example}\label{example:restrictive}
Consider the  rule set $\ruleset = \{r9, r10\}$  where $r9$ and $r10$ are in Tbl.~\ref{tab:examplefig}. Suppose a desired behavior, formalized as $\purposeFact_1$ in Tbl.~\ref{tab:examplefig}, is that after a patient falls, the robot should stay with them for at least 30 minutes.
$\ruleset$ is overly restrictive for $\purposeFact_1$ as $\lang{\purposeFact_1} \bigcap \lang{\ruleset}$ is empty. Therefore, the system cannot execute the desired behavior in $\purposeFact_1$ while respecting the rules in $\ruleset$.
\end{example}
\begin{definition}[Insufficiency]\label{def:sufficiency}
    Given a rule set $\ruleset$ and an undesirable behavior expressed as a set of facts $\concernFact$, $\ruleset$ are \emph{insufficient for preventing $\concernFact$} if and only if $\lang{\ruleset} \cap \lang{\concernFact} \neq \emptyset$.
\end{definition}
\noindent
Intuitively, a set of rules are insufficient if  adhering to those rules can still allow executing an  undesirable behavior.
\begin{example}\label{example:inssuficent}
Consider the \dsl rule $r5$ and the privacy concern $\fact_{1}$ shown in Tbl.~\ref{tab:examplefig}.
Let $\sigma_2$ be the trace shown in Fig.~\ref{fig:traces}, with states $(\{{\small\texttt{OpenCurtainRequest}}\},$ $\measureassign{1}, 1),$ $(\{{\small\texttt{OpenCurtain}}\}, \measureassign{2}, 3)$ where the measures of $\measureassign{1}$ evaluate $\prop$ to $\bot$.
Since $\sigma_2$ is in $\lang{r5} \cap \lang{\fact_{1}}$, then $\lang{r5} \cap \lang{\fact_{1}}$ is not empty.
Therefore, \{$r5$\} is insufficient to prevent $\fact_{1}$, i.e., \{$r_5$\} allows a violation of the privacy concern.
\end{example}

\vspace{-0.08in}
\boldparagraph{Situational Conflicts}
A \dsl rule may interact with other rules in a contradictory manner under specific situations, which can lead to conflicts between different rules. 
However, the definition of \textit{vacuous conflict} (see Def.~\ref{def:dconflict}) does not capture such \textit{situational conflicts} since these conflicts arise only in \emph{particular} situations.  
If the conflicting situations are feasible, then the stakeholders need to resolve the conflict.
We define the notion of \textit{situational conflict} as a generalization of a vacuous conflict.  
%
%
\begin{definition}[Situation]\label{def:situation}
    For a rule $\srule = \rulesyntax{\event \; \andS \; \prop}{\dobg}$, an \emph{$\srule$-triggering situation} is a tuple $(\fos^k_0, \Vec{\measures}_{k})$ where $\fos^k_0 = (\eventocc{0}, \measureassign{0}, \timestamp{0})$ $\ldots$ $(\eventocc{k}, \measureassign{k}, \timestamp{k})$ is the prefix of a trace up to and including state $k$,  $\Vec{\measures}_{k} = \measureassign{k} \ldots \measureassign{n}$ are  \textit{partial} measures (i.e., the functions $\measureassign{n} \in \Vec{{\measures}_k}$ might be undefined for some measures $\measure \in \measures$) for states from $k+1$ onwards, and $\fos^k_0$ triggers $\srule$ in its $k$th state (i.e., $\event \in \eventocc{k} \wedge \measureassign{k}(\prop) = \top$). 
\end{definition}
\noindent
Intuitively, an $\srule$-triggering situation $(\fos_0^k, \Vec{\measures_{k}})$ consists of (1) the complete state information of the past (until $\srule$ is triggered) where all events, measures and their occurrence time are fully observed and (2) the partial measures of the future. Note that the situation cannot observe or control the occurrences of events beyond state $k$ since we want to show that triggering $\srule$ at state $k$ is sufficient to cause a conflict regardless how the system responds after state $k$.  

A situation is \emph{feasible} with respect to a rule set $\ruleset$ if the past up to state $k$ (i.e., $\fos_0^k$) does not \emph{violate} any rule in $\ruleset$. To provide a formal definition of non-violation up to some state ($k$), we extend the semantics of \dsl to evaluate a trace up to a specified state $k$ to check whether a \textit{violation} of rule $\srule$ has already occurred, denoted as $\fos \vdash^k \srule$. We denote the extended semantics as the \emph{bounded semantics}, which is presented in Appendix~\ref{ap:extendedsem}. 
\begin{example}
    Consider the rule $r5$ in Tbl.~\ref{tab:examplefig} and the trace $\sigma_3$ = $({{\small\texttt{OpenCurtainRequest}}}, \measureassign{1}, 1)$ 
    shown in Fig.~\ref{fig:traces} where $\measureassign{1}({\small\texttt{underDressed}})$ is $\bot$. The trace $\sigma_3$ is feasible  with respect to $r5$ because it does not violate $r5$ before state where $r5$ is triggered ($\fos \models^{1} r5$), but it violates $r5$ after 30 minutes ($\fos \not\models r5$) since the response (\texttt{OpenCurtain}) does not occur within 30 minutes after state 1. 
\end{example}
\noindent
A rule $\srule$ is situationally conflicting if there exists a feasible $\srule$-triggering situation that eventually causes a conflict.  
\begin{definition}[Situational Conflict]\label{def:conditionconflictfull}
Let a rule set $\ruleset$ be given. A rule $\srule$ in $\ruleset$ is \emph{situationally conflicting} if there exists an $\srule$-triggering \textit{situation} $(\fos^k_0, \Vec{\measures_{k}})$ such that:
(1) $\fos^k_0 \vdash^k \ruleset$, and
(2) there \emph{does not} exist a situation $\fos = (\eventocc{1}, \measureassign{1}, \timestamp{1}), \ldots (\eventocc{n}, \measureassign{n}, \timestamp{n})$ as an extension of $\fos^k_0$,
such that $\fos$ preserves the measures in $\Vec{\measures_{k}}$, i.e., for every $i > k$, and every measure $\measure \in \measures$, either $\measureassign{i}(\measure) = \Vec{\measures_k}[i-k](\measure)$ or  $\Vec{\measures_k}[i-k](\measure)$ is undefined,
and $\fos$ satisfy $\ruleset$ ($\fos \in \lang{\ruleset}$).
\end{definition}


\begin{example}\label{example:condconf}
    Consider the rule set $\{r5', r11\}$ with  $r5'$ and $r11$  in Tbl.~\ref{tab:examplefig}. 
    Let $\fos_0^2$ be the trace shown in Fig.~\ref{fig:traces}, where the measures of the second state evaluate ${\small\texttt{underDressed}}$ to $\bot$. For any partial measure ($*$), $(\fos_0^2, *)$ is an $\srule5'$-triggering situation because 
     $r5'$ is triggered at time point 2, and not violated when triggered (i.e., $\fos_0^2 \vdash^2 r5'$). However, at time point 2, the response of $r5'$ (\notS{} {\small\texttt{OpenCurtain}} \textbf{within} $30$ \textbf{minutes}) conflicts with the remaining response of $r11$ ({\small\texttt{OpenCurtain}} \textbf{within} $40-(11-1)$ \textbf{minutes}).  Thus, for any future partial measure $(*)$, $r5$ is situationally conflicting w.r.t.  ($\fos_0^2$, $*$).
\end{example}
\vspace{-0.1in}
\begin{remark}
    A rule $\srule$ is vacuously conflicting, if and only if it is situationally conflicting for every $r$-triggering situation ($\fos_0^k$, $\Vec{\measures_k}$).
\end{remark}

\vspace{-0.1in}
\section{\upCRRI and Satisfiability} \label{sec:sanitization}

In this section, we present  our approach, \approach (see Fig.~\ref{fig:overview}), for checking  \nNFRs~\CRRI via \fol satisfiability checking. 
We first show that most of well-formedness checks (except situational conflicts) can be reduced to behavior emptiness checking (Sec.~\ref{ssec:emptycheck}). 
Then, we translate the emptiness checking problem into \fol satisfiability (Sec.~\ref{ssec:translate}) to leverage the SoTA \fol satisfiability checker. 
Finally, we show that the identification of situational conflicts can also be reduced to satisfiability checking (Sec.~\ref{ap:translateSituation}).

\subsection{\upCRRI via Emptiness Checking}\label{ssec:emptycheck}
In this section, we show that the task of identifying vacuous conflict (Def.~\ref{def:dconflict}), redundancy (Def.~\ref{def:redudent}), restrictiveness (Def.~\ref{def:restrictivness}), and insufficiency (Def.~\ref{def:sufficiency}) for a given rule set can be reduced to behavior emptiness checking on the behavior of the rule set. Formally:
\begin{definition}[Rule Emptiness Check]\label{def:emptiness}
Let a rule set $\ruleset$ and a set of facts $\facts$ be given. 
The \emph{emptiness check} of   $\ruleset$ with respect to $\facts$ is $\lang{\ruleset} \cap \lang{\facts} = \emptyset$.
\end{definition}
The input to emptiness checking is a set of rules $\ruleset$ and a set of facts $\facts$. In the following, we describe the appropriate inputs to enable emptiness checking for identifying \nNFRs{} \CRRI. 

\boldparagraph{Detecting Vacuous Conflict}
A rule $\srule \in \ruleset$ is vacuously conflicting in $\ruleset$ if the behavior defined by $\ruleset$ ($\lang{\ruleset}$) does not include a trace where $\srule$ is triggered. 
Therefore, we can determine if a rule is vacuously conflicting by describing the set of traces where $\srule$ is triggered as facts in \dsl, and then checking whether it intersects with $\lang{\ruleset}$ via emptiness checking (Def.~\ref{def:emptiness}). 
\begin{definition}[Triggering]\label{def:ruletrigger}
    Let a rule $\srule = \rulesyntax{\event \; \andS \; \prop}{\dobg}$ be given. The \emph{triggering} of $\srule$, denoted as $\triggerfunc{\srule}$, is the fact $\fact = \factsyntax{\event \; \andS \; \prop}{\top}$.
\end{definition}

\begin{lemma}\label{lemma:vcdef}
    Let a set of rules $\ruleset$ be given. A rule $\srule \in \ruleset$ is vacuously conflicting if and only if $\lang{\ruleset} \cap \lang{\triggerfunc{\srule}}$ is empty.
\end{lemma}

Lemma~\ref{lemma:vcdef} is a direct consequence of Def.~\ref{def:dconflict} and Def.~\ref{def:ruletrigger}.

\boldparagraph{Detecting Redundancies}
Recall from Def.~\ref{def:redudent}, that a rule is redundant in a given rule set $\ruleset$ if and only if removing the rule from $\ruleset$ does not affect the behavior of $\ruleset$.

\vspace{-0.03in}
\begin{lemma}\label{lemma:nonredundant}
    Let a set of rules $\ruleset$ be given. A rule $\srule$ is not redundant if and only if there exists a trace $\fos$ such that $\fos \models \ruleset \setminus \srule$ and $\fos \not\models \srule$. 
\end{lemma}

\vspace{-0.05in}
Lemma~\ref{lemma:nonredundant} is a direct consequence of Def.~\ref{def:redudent}. 
To take advantage of this lemma, we define a notion of \emph{noncompliance} to capture $\fos \not\models \srule$. 
\begin{definition}[Noncompliant]\label{def:violation}
    Let a rule $\srule = \rulesyntax{\event \; \andS \; \prop}{\dobg}$ be given. The \emph{noncompliant fact of $\srule$}, denoted  as $\noncompliant{\srule}$, is the fact $\fact = \factsyntax{\event \; \andS \; \prop}{\notS \; (\dobg})$. 
\end{definition}
\begin{example}
     The noncompliant fact derived from $r5$ (see Tbl.~\ref{tab:examplefig}), is $\noncompliant{r5}$ $=$
    \textbf{exists} {\small\texttt{OpenCurtainRequest}} \textbf{and} \textbf{not} {\small\texttt{underDressed}} \textbf{while} \textbf{not} {\small\texttt{OpenCurtain}} \textbf{within} 30 \textbf{min}.
\end{example}

\vspace{-0.05in}
\begin{lemma}\label{lemma:concern}
    Let a rule $\srule = \rulesyntax{\event \wedge \prop}{\dobg}$ be given. For a trace $\fos$,  $\fos \not\models \srule$ if and only if $\fos \models \noncompliant{\srule}$.  
\end{lemma}

\vspace{-0.05in}
The proof of Lemma~\ref{lemma:concern} follows directly from the semantics of normalized \dsl described in Sec.~\ref{sec:background}.
\begin{lemma}\label{lemma:redundantdef}
    Let a set of rules $\ruleset$ be given. A rule $\srule \in \ruleset$ is redundant if and only if $\lang{\ruleset \setminus \srule} \cap \lang{\noncompliant{\srule}}$ is empty.
\end{lemma}

\vspace{-0.03in}
Lemma~\ref{lemma:redundantdef} is a direct consequence of Def.~\ref{def:redudent}, Lemma~\ref{lemma:nonredundant}, and Lemma~\ref{lemma:concern}.

\boldparagraph{Detecting Restrictive or Insufficient \nNFRs}
Definitions of restrictiveness (Def.~\ref{def:restrictivness}) and insufficiency (Def.~\ref{def:sufficiency}) are already expressed as  emptiness checking problems.


\subsection{Encoding SLEEC DSL Rules \& Facts in \fol}\label{ssec:translate}

To check emptiness (Def.~\ref{def:emptiness}) of $\ruleset$ with respect to $\facts$, we can check the satisfiability of the \fol constraint $\translate{\ruleset} \wedge \translate{\facts}$.  

The signature of \dsl is mapped to the signature of~\fol as follows: (1) every event $\event \in \events$ is mapped to a class of  relational objects $\Class{\event}$; and (2) every measure $\measure \in \measures$ is mapped to a named attribute in relational objects of class $\Class{\measures}$, where $\Class{\measures}$ is a special class of relational objects for measures. 
The function $T$ then translates every \dsl rule $\srule $ to a \fol formula $\translate{\srule}$ such that $\lang{\srule}$ is not empty if and only if $\translate{\srule}$ is satisfiable. $T$ follows the semantics of normalized $\dsl$ detailed in Sec.~\ref{sec:background} to translate  the exact condition for \textit{fulfillment} and \textit{violation}.  

To translate the fulfillment relationship of \dsl responses with respect to a time point $i$ (i.e., $\models_i$), $T$ uses a helper function $T^*$ (see Appendix~\ref{ap:translate}) to recursively visit the elements of the \dsl rule (i.e., obligations, propositions and terms) and encode the condition of fulfillment with respect to the time point $i$ (where the time value and measures of state $i$ are represented by the input relational object $\obj{\measures}$).  
$T$ also uses a helper function \textbf{Violation} (see Appendix~\ref{ap:translate}) to translate the condition for obligations and conditional obligations to be violated at some time point $j$.  
Specifically, $T$ translates a rule
$\rulesyntax{\event \wedge \prop}{\dobg}$ into the constraint ``whenever the triggering event $\event$ occurs ($\forall \obj{\event}:\Class{event})$, the measures are available 
(i.e., $\exists \obj{\measures}:\Class{\measures}\cdot \obj{\event}.\textit{time} = \obj{\measures}.\textit{time}$), and if the measures 
 satisfy the triggering condition ($T^{*}(\prop, \obj{\measures})$), then the obligation chain should be fulfilled ($T^{*}(\dobg, \obj{\measures})$)''. 
Similarly, every fact is translated into the constraint ``there exists an occurrence of an event $\event$ ($\exists \obj{\event}:\Class{event})$, while the measures are available 
(i.e., $\exists \obj{\measures}:\Class{\measures}\cdot \obj{\event}.\textit{time} = \obj{\measures}.\textit{time}$) where the condition $\prop$ is satisfied ($T^{*}(\prop, \obj{\measures})$), and the obligation chain is fulfilled ($T^{*}(\dobg, \obj{\measures})$)''. 

To ensure that no two different measures for the same time point exist, we introduce the \fol \textit{measure consistency axiom}: $\textit{axiom}_{mc} \gets \forall \obj{\measures}_1, \obj{\measures}_2: \Class{\measures} (\obj{\measures}_1.\textit{time} = \obj{\measures}_2.\textit{time} \Rightarrow \obj{\measures}_1 \equiv \obj{\measures}_2)$, where $\equiv$ is the semantically equivalent relational constraint that asserts equivalence of every attribute between $\obj{\measures}_1$ and $\obj{\measures}_2$.

We prove the correctness of the \fol translation and provide the detailed translation function $T$ in Appendix~\ref{app:sanitizationSoundness}.



   
\vspace{-0.03in}
\begin{example}\label{example:translate}
    The \fol translation of  rule $r5$ is $\forall \obj{cr}:$ $\Class{{\small\texttt{OpenCurtainRqst}}} . $ $(\exists\obj{\measures}:\Class{\measures}  . 
        (\obj{cr}.\textit{time}=\obj{\measures}.\textit{time} \wedge \neg \obj{\measures}.{\small\texttt{underDressed}})) \Rightarrow$  $\exists \obj{co}:\Class{{\small\texttt{OpenCurtain}}} \; . (\obj{cr}.\textit{time} \le \obj{co}.\textit{time} \le \obj{cr}.\textit{time} + 30)$.
\end{example}

\begin{example}
 Consider the rule set $\ruleset = \{r5 , r8\}$. 
 $\triggerfunc{r5}$ is \textbf{exist} {\small\texttt{OpenCurtainRequest}} \andS{} \notS{} {\small\texttt{underDressed}} \textbf{while} $\top$. To check whether $\srule_5$ is vacuously conflicting, we reduce the emptiness problem for $\lang{\ruleset} \bigcap \lang{\triggerfunc{r5}}$ to the satisfiability of 
  \fol formula $\translate{r5} \wedge \translate{r8} \wedge \translate{\triggerfunc{r5}} \wedge \textit{axiom}_{mc}$. The \fol formula is UNSAT which means that the set of all traces where $r5$ is triggered does not intersect with the behaviour allowed by $r8$, and thus $r5$ is vacuously conflicting in $\ruleset$ due to $r8$.  
\end{example}

\begin{example}\label{example:redudentSAT}
    Consider the rule set $\ruleset = \{r5, r6, r7\}$. To check whether $r5$ is redundant, we reduce the emptiness checking problem for $\lang{\ruleset \setminus r5} \bigcap \lang{\noncompliant{r5}}$ to  satisfiability of the \fol formula $\translate{r8} \wedge \translate{\noncompliant{r5}} \wedge \textit{axiom}_{mc}$. The \fol formula is UNSAT thus the set of traces not complying with $r5$ ($\noncompliant{\ruleset}$) does not intersect with the behaviour of $r6$ and $r7$. $r5$ is logical consequence of $r6$ and $r7$, and so $\ruleset$ is redundant.
\end{example}
\begin{example}
    Consider the \dsl rule $r5$ and  the fact  $\fact_{1}$ in Tbl.~\ref{tab:examplefig} for describing the desirable behavior ``a user requests to open the curtain while she is not under-dressed''. The \fol formula $\translate{r5} \wedge \translate{\fact_{1}}$ has a satisfying solution $(\domain, v)$ where (1) $\domain = \{\obj{\measures}_1, \obj{\measures}_2,$ $\obj{{\small\texttt{OpenCurtainRqst}}}, \obj{\texttt{OpenCurtain}}\}$; (2) $v(o.\textit{ext} = \top)$ for every $o \in \domain$; and (3) $v(\obj{\measures}_1.{\small\texttt{underDressed}}) = \bot \wedge v(\obj{\measures}_1.{\small\texttt{time}}) = 1 \wedge$
    $v(\obj{\measures}_2.{\small\texttt{time})} = 2$. The solution is mapped to a trace $\sigma_5 = (\{{\small\texttt{OpenCurtainRequest}}\}, \measureassign{1}, 1),$ $(\{{\small\texttt{OpenCurtain}}\}, \measureassign{2}, 3)$ where $\measureassign{1}({\small\texttt{underDressed}}) = \bot$, shown in Fig.~\ref{fig:traces}. The trace $\fos$ is contained in $\lang{r5} \cap \lang{\fact_{1}}$.
    Therefore,  $r5$ does not restrict the behavior defined in $\fact_{1}$.
\end{example}


\subsection{Situational Conflicts via Satisfiability}\label{sec:situationConflicts}
Recall from Def.~\ref{def:conditionconflictfull} that the detection of situational conflict involves searching for a situation (i.e., a partial trace) such that all of its extensions are rejected by a given rule set. This corresponds to an \textit{exists} (situation)-\textit{forall} (extension) query, which makes a direct reduction to emptiness checking, defined in Sec.~\ref{ssec:emptycheck}, difficult. 
Fortunately, in the \dsl language, the time limit of obligations is always defined as \textit{continuous} time windows. This characteristic allows us to identify conflicts when a positive window (an event $\event$ must occur) is completely covered by the continuous union of negative windows (the event $\event$ must not occur), and
detect a subset of situational conflicts by checking the reachability of a state where a rule is triggered while its obligations are blocked. 

In this section, we propose a sound approach, through a direct reduction to \fol satisfiability, to detect situational conflicts. Note that we do not tackle
classes of situational conflicts without blocked responses, called \emph{transitive situational conflicts} (See Remark ~\ref{remark:tanssc} in Appendix~\ref{ap:translate}).
We start by  an informal discussion of the different states of obligations: triggered, fulfilled, violated, active, forced, and blocked.  The relationship between them is graphically displayed in  Fig.~\ref{fig:obgdep}. The formal definitions are in Appendix~\ref{ap:translateSituation}.

\begin{figure}
    \centering
\includegraphics[width=0.4\textwidth]{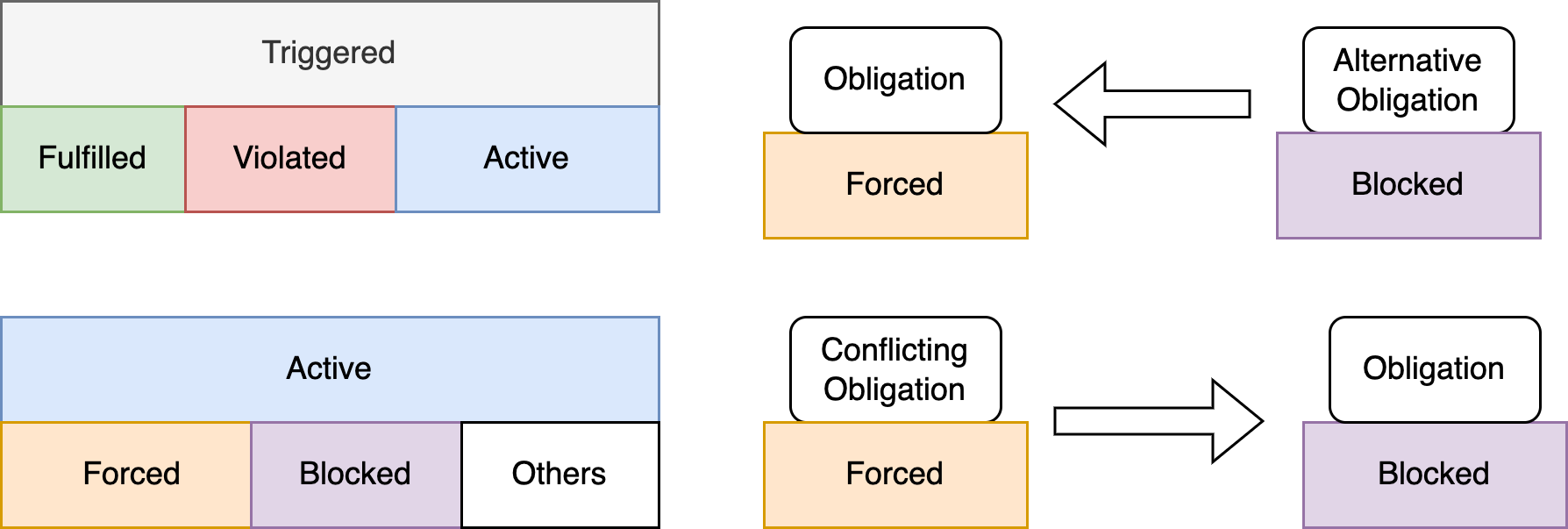}
    \vspace{-0.1in}
    \caption{\small{The relationship between different states of obligations.}}
    \vspace{-0.2in}
    \label{fig:obgdep}
\end{figure}

\vspace{-0.03in}
\noindent
\boldparagraph{States of Obligations}
Given a situation $(\fos^k_0, \measures_k)$, an obligation $\obg$ is \textit{triggered} if some states in $\fos^k_0$ trigger a rule $\srule$ where $\obg$ is the response of $\srule$. Once $\obg$ is triggered, it might be \textit{fulfilled} or \textit{violated} in $\fos^k_0$. However, if $\obg$ is neither fulfilled nor violated, the obligation is then \textit{active} which can be fulfilled or violated in the future (in the extension of $\fos_0^k$) in order to satisfy all the rules. The relationship between triggered, fulfilled, violated and active obligations is visualized on the left hand side of Fig.~\ref{fig:obgdep}.

\vspace{-0.05in}
\begin{example}\label{example:sc1}
    Consider the rule set $\{r5', r11\}$ where $r5'$ and $r11$ are  in Tbl.~\ref{tab:examplefig}. 
    Let $\fos_0^2$ be a trace prefix, with 2 states, shown in Fig.~\ref{fig:traces}, where the measures of the second state evaluate \texttt{underDressed} 
    to $\bot$. In the situation ($\fos_0^2, *$), $r5'$ is triggered at time point 2, and its response, $\obg =$ \notS{} \texttt{OpenCurtain} \textbf{within} $30$ \textbf{min}, is triggered as well. Since event \texttt{OpenCurtain} does not occur at time point 2, $\obg$ is not violated but is also not fulfilled.  
     Therefore, $\obg$ is an active obligation at time point 2.  At time point 1, $r11$ is triggered, and its response, $\obg'=$ \texttt{OpenCurtain} \textbf{within} $40$ \textbf{min}, is triggered as well. 
     Since \texttt{OpenCurtain} does not occur at either time point 1 or 2, $\obg'$ is neither fulfilled nor violated. Thus, $\obg'$ is active at time point 2.
\end{example}

\vspace{-0.05in}
An active obligation can be in two special states:
\textit{forced} and \textit{blocked}. An active obligation $\obg$ is \textit{forced} if it is the only remaining option for an active response, i.e., there is no alternative to $\obg$ that is not blocked.  Therefore, violating a forced obligation would inevitably cause a violation of some rules. An obligation $\obg$ is blocked if there exists a ``conflicting'' obligation $\obg'$, and $\obg'$ is forced.  The dependencies between forced and blocked obligations are visualized in the right hand side of Fig.~\ref{fig:obgdep}. When a rule is triggered while its response is blocked in ($\fos^k_0$, $\Vec{\measures}_{k}$), a conflict is inevitable. 
\begin{example}\label{example:sc2}
    Continuing from Ex.~\ref{example:sc1}, let $\obg$ and $\obg'$ be the two active obligations triggered by $\fos_0^2$ at time points 2 and 1, respectively.
    Since time point 2 occurs 10 min. after time point 1 in $\fos^k_0$, the remaining time limit of the obligation $\obg'$ is $30$ min. Both $\obg$ and $\obg'$ are forced since they do not have alternatives. Both $\obg$ and $\obg'$ are also blocked at time point 2 since they are conflicting.  
\end{example}
There are circular dependencies between forced and blocked obligation chains! Fortunately, by encoding the status definitions into \fol, we can leverage the \fol solver's ability to incrementally and lazily unroll the necessary definitions to resolve such dependencies (see Appendix~\ref{ap:translateSituation} for details) 

\vspace{-0.03in}
\begin{lemma}\label{lemma:sc}
     For every rule $\srule = \rulesyntax{\event \wedge \prop}{\dobg}$ in a rule set $\ruleset$, if there exists a $\srule$-triggering situation ($\fos^k_0$, $\Vec{\measures}_{k}$) (see Def.~\ref{def:situation}) 
    where the obligation chain $\dobg$ is blocked at time point $k$, 
    then $\srule$ is situationally conflicting with respect to the situation ($\fos^k_0$, $\Vec{\measures}_{k}$).
\end{lemma}

\vspace{-0.05in}
Lemma~\ref{lemma:sc} enables the detection of situational conflicts for a rule $\srule$ in  $\ruleset$ by encoding the sufficient condition for situational conflict into \fol (denoted by $\sctranslate{\ruleset, \srule}$).
Every satisfying solution to $\sctranslate{\ruleset, \srule}$  represents a situation where $\srule$ is situational conflicting.  
The proof of Lemma~\ref{lemma:sc}, the situational encoding $\sctranslate{\ruleset, \srule}$ and its correctness proof sketch are in Appendix~\ref{ap:translateSituation}.

\subsection{The \approach Approach} 
\approach turns well-formedness problems into \fol satisfiability checking using Lemma~\ref{lemma:vcdef} for vacuous conflicts, Lemma~\ref{lemma:sc} for situational conflicts, Lemma~\ref{lemma:redundantdef} for redundancies, and Def.~\ref{def:restrictivness} and \ref{def:sufficiency} for restrictiveness and  insufficiency, respectively. Following this, \approach encodes the input \dsl rules  into \fol (Step 1 in Fig.~\ref{fig:overview}), as explained in Sec.~\ref{ssec:translate}. 
Then, it identifies \nNFR{} \CRRI using a SoTA \fol satisfiability checker (Step 2). The remaining step, computing the diagnosis, is described in Sec.~\ref{sec:diagnose}.


\section{Diagnosis for \upCRRI}\label{sec:diagnose}
\approach not only detects \CRRI  but also diagnoses their causes, 
thereby aiding stakeholders in comprehending and debugging problematic rules. 
Below, we present the diagnosis process for different well-formedness problems, leveraging the result of satisfiability checking. 
We first show how to utilize the causal proof  of (un)satisfiability (Sec.~\ref{sec:proof}) to pinpoint the reasons behind redundancy, vacuous-conflict, and restrictiveness (Sec.~\ref{ssec:diagUnsat}). 
Then we present a satisfying \fol solution-based diagnosis for insufficiency and  
combine the solution-based and proof-based diagnosis to highlight the cause of situational conflict (Sec.~\ref{ssec:diagHybrid}).

\subsection{Causal \fol Proof}
\label{sec:proof}
We first introduce causal \fol proofs and  \fol derivation rules. Then we focus on the \textbf{Impl} rule and discuss its role as the \textit{core} of UNSAT derivations for computing our diagnosis.
\emph{A causal \fol proof} is a sequence of derivation steps $L_1, L_2 \ldots, L_n $. Each  step $L_i$ is a tuple ($i$, $\psi$, $o$, $\textit{Deps}$ ) where (1) $i$ is the ID of the derivation step; 
(2) $\psi$ is the derived \fol lemma; (3) $o$ is the name of the \textit{derivation rule} used to derive $\psi$; and (4) $\textit{Deps}$ are  IDs of dependent lemmas for deriving  $\psi$. A \emph{derivation step is sound} if the lemma $\psi$ can be obtained via the derivation rule $o$ using lemmas from $\textit{Deps}$. A \emph{proof is sound} if every derivation step is sound. The proof is \textit{refutational} if the final derivation contains the lemma $\bot$ (UNSAT).  

\begin{table}
    \centering
    \caption{\small \fol causal proof of UNSAT for  $\phi_1$ and $\phi_2$ in Ex.~\ref{example:proof1}. }
    \vspace{-0.1in}
    \scalebox{0.78}{
        \begin{tabular}{c c c c}
        \toprule
             ID & Lemma & Derivation Rule & Deps \\
             \toprule
             1 & \makecell{$\forall a:A\exists b:B\cdot (a.\textit{time} \le b.\textit{time} \le a.\textit{time} + 10)$} & \textbf{INPUT}: $\phi_1$ & \{\}\\
             2 & \makecell{$\exists a:A\forall b:B \cdot(b.\textit{time} \ge a.\textit{time} + 20 \wedge p(a,b))$} &  \textbf{INPUT}: $\phi_2$ & \{\}\\
             3 & \makecell{$\forall b:B \cdot(b.\textit{time} \ge a_1.\textit{time} + 20  \wedge p(a_1,b))$} & \textbf{EI}: $[a \gets a_1]$ & \{2 \}\\
             4 & \makecell{$\exists b:B\cdot (a_1.\textit{time} \le b.\textit{time} \le a_1.\textit{time} + 10)$} &  \textbf{UI}: $[a \gets a_1]$ & \{1, 3\}\\
             5 & \makecell{$a_1.\textit{time} \le b_1.\textit{time} \le a_1.\textit{time} + 10$} & \textbf{EI}: $[b \gets b_1]$ & \{4\}\\
             6 & \makecell{$b_1.\textit{time} \ge a_1.\textit{time} + 20  \wedge p(a_1,b_1)$} & \textbf{UI}:[$b \gets b_1$]  & \{3, 5\}\\
             7 & \makecell{$b_1.\textit{time} \ge a_1.\textit{time} + 20$} & \textbf{And}& \{6\}\\
             8 & \makecell{$\bot$} & \textbf{Impl}& \{5, 7\}\\
        \bottomrule
        \end{tabular}
    }
    \label{fig:simpleproofexample}
\end{table}

\begin{example}\label{example:proof1}  
    Let \fol formulas, $\phi_1: \forall a:A\exists b:B\cdot (a.\textit{time} \le b.\textit{time} \le a.\textit{time} + 10)$ and $\phi_2: \exists a:A\forall b:B \cdot(b.\textit{time} \ge a.\textit{time} + 20 \wedge p(a,b))$ be given, where $A$ and $B$ are classes of relational objects, \textit{time} is an attribute of type $\mathbb{N}$ and $p$ is a complex predicate. $\phi_1 \wedge \phi_2$ is UNSAT, and the proof of UNSAT is shown in Tbl.~\ref{fig:simpleproofexample}. The proof starts by introducing the \textbf{Input} $\phi_1$ and $\phi_2$ (steps 1, 2), and then uses existential instantiation (\textbf{EI}) in steps 3, 5 and universal instantiation (\textbf{UI})  in steps 4, 6 to eliminate quantifiers in $\phi_1$ and $\phi_2$. In step 7, the \textbf{And} rule decomposes the conjunction and derives part of it as a new lemma. 
    Finally, in step 8, $\bot$ is derived with the \textbf{Impl} rule from the quantifier-free (QF) lemmas derived in steps 5 and 7. The proof is refutational because step 8 derives $\bot$.
 \end{example}

 \vspace{-0.05in}
 Given a refutational proof, we can check and trim the proof following a  methodology similar to \cite{DBLP:conf/fmcad/HeuleHW13}. In the rest of the section, we assume every proof is already trimmed, and provide the detail on proof checking and trimming in Appendix~\ref{ap:treereduction}. 

In a refutational proof, the final derivation of $\bot$ must be carried out by the \textbf{Impl} rule (e.g., step 8 in the proof in Tbl.~\ref{fig:simpleproofexample}). Given a proof $L$, $\textit{Imp}(L)$ is the set of QF lemmas derived via or listed as dependencies for \textbf{Impl}. Intuitively, $\textit{Imp}(L)$ is the set of \emph{core} QF formulas derived from the inputs following other derivation rules and sufficient to imply the conflict (UNSAT).  By analyzing lemmas in $\textit{Imp}(L)$ and backtracking their dependencies, we determine which
parts of inputs  are used to derive UNSAT. 

\vspace{-0.03in}
\begin{example}\label{example:proof3}
    Continuing from Ex.\ref{example:proof1}, the final step (step 7) of the proof $L$ derives $\bot$ using the \textbf{Impl} rule with dependency $\{5, 7\}$. Therefore, $\textit{Imp}(L)$ is the lemmas derived at steps 5 and 7,  and the term $p(a)$ in $\phi_2$ is not involved for deriving these lemmas and thus not necessary for the UNSAT conclusion. Therefore, the proof remains valid by removing $p(a)$ from $\phi_2$. 
\end{example}


\vspace{-0.1in}
\subsection{Diagnosis for Redundancy, Vacuous and Situational Conflict, and Restrictiveness}
\label{ssec:diagUnsat}
In addition to reporting a binary ``yes'' or ``no'' answer to the satisfiability, the \fol satisfiability checker \lego also provides a causal proof of UNSAT (see   Sec.~\ref{sec:proof}) if the encoded formula is unsatisfiable.  We project the proof into the input \dsl rules to highlight the causes of \CRRI problems at the level of \textit{atomic elements}. 

\begin{definition}[Atomic element]\label{def:ap}
An \emph{atomic element} in \dsl is one the followings: (1) an \textit{atomic proposition} $\aprop$:   $\top \mid \bot \mid \term = \term \mid \term \ge \term \mid \neg \aprop$, (2) a triggering event $\event$ (where $\event$ in $\rulesyntax{\event \wedge 
 \ldots}{\ldots}$ or $\factsyntax{\event \wedge 
 \ldots}{\ldots}$), (3) a response event $\event'$ in an obligation (in $\within{\event'}{\ldots}$), or (4) a deadline $\term$ of an obligation (in $\within{\ldots}{\term}$).
\end{definition}
We now informally define \emph{involved atomic elements}
~\footnote{
A formal definition  is in Appendix~\ref{ap:iae}.
}.
Let a causal proof $L$ and the subset of lemmas derived by \textbf{Impl} $\textit{Imp}(L)$ be given. An atomic element $\aelement$ is \textit{involved in the derivation of UNSAT}, denoted as $inv(\aelement, L)$, if the atomic element encodes the necessary dependency of lemmas in $\textit{Imp}(L)$.
 An \emph{atomic proposition $\aprop$ is involved} if it is used to encode some QF lemma in $\textit{Imp}(L)$. 
\emph{A triggering event $\event$ is involved} if $\textit{Imp}(L)$ contains a QF implication where the premise of implication is $\obj{\event}.ext$ 
($\obj{\event} \Rightarrow \ldots$). A \emph{response event is involved} if $\textit{Imp}(L)$ contains  a QF implication where the conclusion of implication is $\obj{\event}.ext$ 
($\ldots \Rightarrow \obj{\event}$). A \emph{deadline $\term$ 
is involved} if there is a lemma in $\textit{Imp}(L)$ that contains the time constraint encoded from $\term$.
\begin{example}
    Following Ex.~\ref{example:proof3}, the proof of UNSAT $L$ for $\phi_1 \wedge \phi_2$ is shown in Tbl.~\ref{fig:simpleproofexample}, and $\phi_1 \wedge \phi_2$ is the encoding for redundancy checking of $r13$ in the rule set $\{r13, r14\}$ where $r13$ and $r14$ are in Tbl.~\ref{tab:examplefig}:    $\phi_1 = T(r13)$ and $\phi_2 = T(\noncompliant{r14})$. From the proof $L$ and $\textit{Imp}(L)$, the involved atomic elements in $r13$ are  $A$ and $\within{B}{5\;\textbf{min}}$. The involved atomic elements in $r14$ are  $A$ and  $\within{B}{10\;\textbf{min}}$. The involved elements indicate the reason of redundancy: $r13$ and $r14$ have the same trigger $A$, and the response of $r14$ ($\within{B}{20 \; \textbf{min}}$) is logically implied by the response of $r13$ ($\within{B}{20 \; \textbf{min}}$). The element $p(A)$ in $r14$ is not involved, and the redundancy does not depend on it. 
\end{example}

\vspace{-0.03in}
The diagnosis for redundancy, vacuous-conflict, and restrictiveness presents only those $\dsl$ rules $\srule$ in an UNSAT proof $L$, 
where $L$ derives $\translate{\srule}$ by the derivation rule \textbf{Input}. The diagnosis also highlights every involved atomic element (see Fig.~\ref{fig:unsatDiag}). 
To efficiently identify involved atomic elements for a given proof $L$, the translation function is augmented to keep track of the mapping between input \dsl elements (events, propositions, terms) and the translated \fol elements. For every lemma in $\textit{Imp}(L)$, we can efficiently obtain the source elements in \dsl using the mapping and identify the involved atomic elements. 



\begin{figure}
    \centering
    \shadowbox{\includegraphics[scale=0.2]{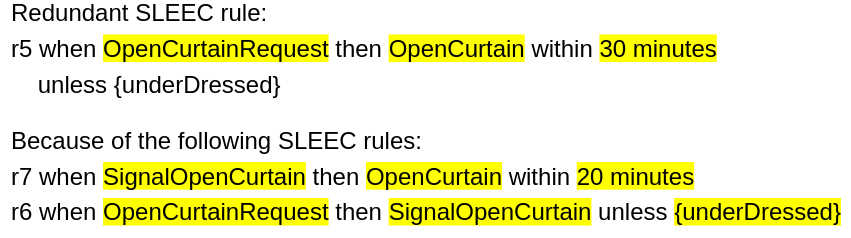}}
    \vspace{-0.045in}
    \caption{{\small Redundancy diagnosis for Ex.~\ref{exmaple:redudent}: a redundant rule $r5$ together with the set of rules $\{r6, r7\}$ with which $r5$ is redundant.  Redundant clauses are highlighted. }}
    \label{fig:unsatDiag}
    \vspace{-0.18in}
\end{figure}

\vspace{-0.05in}
\subsection{Other Diagnoses}
\label{ssec:diagSol}
\label{ssec:diagHybrid}

\noindent
\boldparagraph{Diagnosis for Insufficiency}
A set of rules $\ruleset$ is insufficient to block a concern $\fact$ if the \fol formula $T(\ruleset) \wedge T(\fact)$ 
is satisfiable. The diagnosis of $\ruleset$ insufficiency is   a trace constructed from a solution to the insufficiency  query for $\fact$ and representing one undesirable behavior allowed by $\ruleset$, that realizes $\fact$ (see  Fig.~\ref{fig:satDiag}).


\begin{figure}
    \centering
    \shadowbox{\includegraphics[scale=0.25]{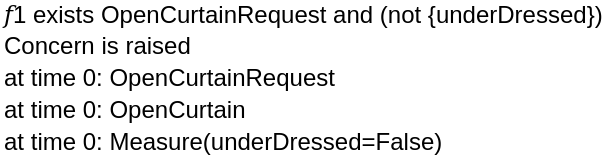}}
    \vspace{-0.1in}
    \caption{{\small Insufficiency diagnosis for  Ex.~\ref{example:inssuficent}: a trace where a concern is raised while respecting the rules.}}
    \label{fig:satDiag}
    \vspace{-0.10in}
\end{figure}


\vspace{-0.04in}
\noindent
\boldparagraph{Diagnosis for Situational Conflict}
A rule $\srule$ is situationally conflicting if there exists a situation under which there is a conflict. 
 The diagnosis needs to capture both the situation and the reason for the conflict.
Therefore, we produce a hybrid diagnosis which combines the \emph{solution-based diagnosis} including a trace representing the situation where the rule is conflicting, and the \emph{causal-proof diagnosis} pinpointing the reasons for the conflicts (using the \nNFR{} rules) given the situation.
\begin{figure}
    \centering
    \shadowbox{\includegraphics[scale=0.18]{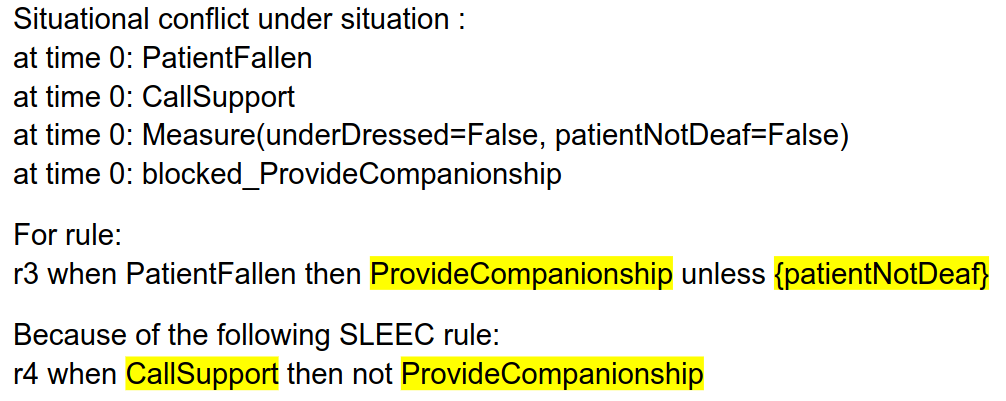}}
    \vspace{-0.15in}
    \caption{{\small Situational conflict diagnosis for Ex.~\ref{example:condconf}:   a trace of the situation where the conflict occurs, the conflicting rule $r3$ and the rule set $\{r4\}$. The  conflicting clauses are highlighted.}}
    \label{fig:satUnsatDiag}
    \vspace{-0.15in}
\end{figure}
To do so,  we first obtain the conflicting situation, and then encode it as facts to check whether the situation causes conflict between the rules. More specifically, given a rule $\srule \in \ruleset$, the situation where $\srule$ is conflicting can be constructed from the \fol solution to the situational-conflict checking query. 
By asserting the situation as facts $\fact_{\textit{sit}}$ and then checking whether  $\lang{\ruleset} \cap \lang{\fact_{\textit{sit}}}$ is empty, we can confirm the conflict and then use causal-proof based diagnosis to highlight the causes of conflict given the situation,
as illustrated in Fig.~\ref{fig:satUnsatDiag}.

\section{Evaluation}
\label{sec:evaluation}

\noindent
\boldparagraph{Tool Support} We implemented our \approach approach in a tool \tool (see the architecture in Fig.~\ref{fig:overview}). Input to \tool are \nNFRs{} expressed in \dsl.  Outputs are diagnoses for each \CRRI issue identified (e.g., Fig.~\ref{fig:unsatDiag}-\ref{fig:satDiag}). Built using Python, \tool encodes \dsl rules into \fol \CRRI satisfiability queries. It then uses \lego~\cite{feng-et-al-23} as the back-end to check \CRRI \fol satisfiability and produces the causal (un)satisfiability proof. Finally, it computes \CRRI diagnoses based on the proof. When no issues are identified, \tool produces a proof of correctness (i.e., absence of \CRRI). \tool also includes a user interface to
allow users to input \dsl rules.
\tool together with the evaluation artifacts is available in \cite{ICSE-Artifact-new}.


\begin{table}[t]
    \centering
    \sffamily
    \caption{{\small Case study statistics.}}
    
    \vspace{-0.15in}
    \scalebox{0.73}{
        \begin{tabular}{rcccccc}
        \toprule
            name & rules (evnt.$-$ msr.) &  defeaters & N-TS$-$TS & domain & stage \\\toprule
            ALMI & $39$ ($41-15$) &  $8$ & $2$ $-$ $2$ & social & deployed\\
             \cellcolor{gray!10}ASPEN & \cellcolor{gray!10} $15$ ($25-18$)& \cellcolor{gray!10}  $5$ & \cellcolor{gray!10} $2$ $-$ $2$ & \cellcolor{gray!10} env. & \cellcolor{gray!10} design\\
            AutoCar & 19 (36-26) &  20 & $1$ $-$ $1$ & transport & design\\
             \cellcolor{gray!10}BSN & \cellcolor{gray!10} 29 ($33$-$31$)& \cellcolor{gray!10}  $4$ & \cellcolor{gray!10} $2$ $-$ $2$ & \cellcolor{gray!10} health & \cellcolor{gray!10} proto.\\
            DressAssist & 31 (54-42) & 39 & $4$ $-$ $1$ & care & proto.\\
             \cellcolor{gray!10}CSI-Cobot & \cellcolor{gray!10} $20$ ($23$-$11$)& \cellcolor{gray!10}  $8$ & \cellcolor{gray!10} $2$ $-$ $1$ & \cellcolor{gray!10} manifacture & \cellcolor{gray!10} proto.\\
            DAISY & $26$ (45-31) &  20 & $3$ $-$ $1$ & health & proto.\\
             \cellcolor{gray!10}DPA & \cellcolor{gray!10} $26$ (28-25)& \cellcolor{gray!10}  $4$ & \cellcolor{gray!10} $1$ $-$ $1$ & \cellcolor{gray!10} education & \cellcolor{gray!10} design\\
            SafeSCAD & $28$ ($29$-$20$)&  $3$ & $2$ $-$ $1$ & transport & proto.\\
        \bottomrule
        \end{tabular}
    }
    \label{tab:caseStudyStats}  
    \vspace{-0.15in}
\end{table}

To evaluate our approach we ask the research question (\textbf{RQ}): How effective is  \tool in detecting {\CRRI}s? 

As discussed in Sec.~\ref{sec:relatedwork}, we cannot provide a quantitative comparison with existing work (SLEECVAL~\cite{Getir-Yaman-et-al-23} and Autocheck~\cite{feng-et-al-23-b}) for WFIs detection because of the semantic difference in the targeted WFIs. Instead, we report on the results of experiments demonstrating the \emph{effectiveness} of our approach  for identifying relevant well-formedness issues on nine real-world case studies. 


\boldparagraph{Models and Methods}
We have divided the authors of the paper into two groups: the \emph{analysis experts} (AEs) who have defined and developed the approach, and the \emph{stakeholders}, further categorized into \emph{technical}  (TSs) and \emph{non-technical}  (N-TSs). We have ensured that there was no overlap between these groups. The stakeholders include: an ethicist, a lawyer, a philosopher, a psychologist, a safety analyst, and three engineers. The TSs identified nine real-world case studies from different stages, ranging from the design phase to deployed systems, as shown in Tbl.~\ref{tab:caseStudyStats}. 

The case studies were as follows. 
(1) ALMI~\cite{Hamilton-et-al-22}: a system assisting elderly and disabled users  in a monitoring or advisory role and with physical tasks; (2) ASPEN~\cite{aspen-23}: an autonomous agent dedicated to forest protection, providing both diagnosis and treatment of various tree pests and diseases; (3) AutoCar~\cite{autocar-20}: a system that implements emergency-vehicle priority awareness for autonomous vehicles; (4) BSN~\cite{Gil-et-al-21}: a healthcare system detecting emergencies by continuously monitoring the patient’s health status; (5) DressAssist~\cite{Jevtic-et-al-19,townsend-et-al-2022}: an assistive and supportive system used to dress and provide basic care for the elderly, children, and those with disabilities; (6) CSI-Cobot~\cite{Stefanakos-et-al-22}: a system ensuring the safe integration of industrial collaborative robot manipulators; (7) DAISY~\cite{daisy-22}: a sociotechnical AI-supported system that directs patients through an A\&E triage pathway; (8) DPA~\cite{Amaral-et-al-22}: a system to check the compliance of data processing agreements against the General Data Protection Regulation; (9) SafeSCAD~\cite{calinescu2021maintaining}: a driver attentiveness management system to support safe shared control of autonomous vehicles.
Case studies were chosen  (i) to test whether our approach scales with case studies involving complex \nNFRs{} with numerous defeaters and time constraints (AutoCar, DAISY, DressAssist); (ii) to examine the benefits of our approach on early-stage case studies (e.g., ASPEN and AutoCar), as compared to those at a more advanced stage (e.g., ALMI, DPA, BSN); (iii) to compare case studies involving many stakeholders (DressAssist, and DAISY) with those having fewer ones (DPA and AutoCar), and (iv) to consider different domains including the environment, healthcare, and transport.

For each case study, 1-2 TSs were paired with 1-4 N-TSs; together they built a set of normative requirements (\nNFRs). They then met to manually review, discuss, and agree on these \nNFRs. 
They specified a total of 233 \nNFRs. The number of \dsl rules ranged from 15 (ASPEN) to 39 (ALMI), and the number of defeaters in these rules ranged from 4 (DPA and BSN) to 39 (DressAssist).  Tbl.~\ref{tab:caseStudyStats} reports 
the number of events (evnt.), measures (msr.), rules (rules), defeaters (\textbf{unless}), and both technical (TS) and non-technical (N-TS) stakeholders  involved in the writing of the \nNFRs.

\begin{table}[t]
    \centering
    \caption{\small  Run-time performance and  \CRRI. True positive (TP) and false negative (FP) issues identified for each  vacuous conflict (v-conf.), situational conflict (s-conf.), redundancy (redund.), restrictiveness (restrict.), and insufficiency (insuff.).  Spurious issues  are in bold.
    }
    \vspace{-0.1in}
    \scalebox{0.7}{
        \begin{tabular}{c c c c c c c}
             \toprule
            \multirow{ 2}{*}{case studies} & v-conf. & s-conf. & redund. & restrict. & insuffi. & time\\
            & (TP - FP) & (TP - FP) & (TP - FP) & (TP - FP) &  (TP - FP) & (sec.)\\ \toprule
           ALMI & 0 - 0 & 3 - 0 & 0 - 0 & 0 - 0 &  1 - \textbf{1} & 30\\
            \cellcolor{gray!10} ASPEN &  \cellcolor{gray!10} 0 - 0 & \cellcolor{gray!10} 3 - 0 & \cellcolor{gray!10} 1 - 0 & \cellcolor{gray!10} 0 - 0 & \cellcolor{gray!10}  5 - 0 & \cellcolor{gray!10} 25.3\\
           AutoCar & 0 - 0 & 4 - 0 & 2 - 0 & 0 - 0 &  9 - 0 & 27.7\\
            \cellcolor{gray!10} BSN & \cellcolor{gray!10} 0 - 0 & \cellcolor{gray!10} 0 - 0 & \cellcolor{gray!10} 0 - 0 & \cellcolor{gray!10} 0 - 0 & \cellcolor{gray!10}  3 - 0 & \cellcolor{gray!10} 46\\
           DressAssist & 0 - 0 & 1 - 0 & 0 - 0 & 0 - 0 &  1 - \textbf{3} & 20.3\\
            \cellcolor{gray!10} CSI-Cobot & \cellcolor{gray!10} 0 - 0 & \cellcolor{gray!10} 0 - 0 & \cellcolor{gray!10} 2 - 0 & \cellcolor{gray!10} 0 - 0 & \cellcolor{gray!10} 6 - \textbf{1} & \cellcolor{gray!10} 25.3\\ 
           DAISY & 0 - 0 & 1 - 0 & 1 - 0 & 0 - 0 &  5 - 0 & 30.4\\
            \cellcolor{gray!10}DPA & \cellcolor{gray!10}  0 - 0  & \cellcolor{gray!10} 0- 0 & \cellcolor{gray!10} 0 - 0 & \cellcolor{gray!10} 0 - 0 & \cellcolor{gray!10} 4 - 0  & \cellcolor{gray!10} 21.4\\
           SafeSCAD & 0 - 0 & 8 - 0 & 2 - 0 & \textbf{2} - 0 & 4 - \textbf{1} & 42.4\\
           \bottomrule
        \end{tabular}}
    \label{tab:identifyIssues}
    \vspace{-0.2in}
\end{table}

\begin{table}[t]
    \centering
    \caption{\small 
    Debugging effort for non-spurious \CRRI{} in Tbl.~\ref{tab:identifyIssues}. Metrics are: \#  issues, \# added events or measures (add D), \# added  (add R), refined, merged, and deleted rules, and \# cases classified as too complex.} 

    \vspace{-0.1in}
    \scalebox{0.7}{
        \begin{tabular}{ccccccccc}
        \toprule
             \multirow{ 2}{*}{case studies} & \multirow{ 2}{*}{total} &  \multicolumn{5} {c} {resolution effort} & \multirow{2}{*}{complex} \\
              &  & add D & add R & refine & merge & delete &  \\  \toprule
             ALMI & $4$ &  1& 4 & 0 & 0 & 1 & - \\ 
              \cellcolor{gray!10} ASPEN &  \cellcolor{gray!10} $9$  & \cellcolor{gray!10} 1 & \cellcolor{gray!10} 3 & \cellcolor{gray!10} 3 & \cellcolor{gray!10} 1 & \cellcolor{gray!10} 2  &  \cellcolor{gray!10}  - \\ 
             AutoCar & 11 & 1 & 20 & 4 & 0 & 1 & - \\ 
              \cellcolor{gray!10} BSN &  \cellcolor{gray!10} 3 & \cellcolor{gray!10} 3 & \cellcolor{gray!10} 3 & \cellcolor{gray!10} 0 & \cellcolor{gray!10} 0 & \cellcolor{gray!10}0  &  \cellcolor{gray!10}  - \\
             DressAssist & $1$ & 1 & 0 & 1 & 0 & 0 & - \\
              \cellcolor{gray!10} CSI-Cobot &  \cellcolor{gray!10} $8$ & \cellcolor{gray!10} 3 & \cellcolor{gray!10}7 & \cellcolor{gray!10} 1 & \cellcolor{gray!10}0 & \cellcolor{gray!10}1 &  \cellcolor{gray!10}  - \\ 
             DAISY & $7$ & 0 & 5 & 1 & 1 & 0 & - \\ 
              \cellcolor{gray!10} DPA &  \cellcolor{gray!10} $4$ & \cellcolor{gray!10}0 & \cellcolor{gray!10} 4 & \cellcolor{gray!10}1 & \cellcolor{gray!10}0 & \cellcolor{gray!10} 2 &  \cellcolor{gray!10}  - \\ 
             SafeSCAD & $16$ & 5 & 6 &  9 & 4 & 2 & \textbf{1} \\
        \bottomrule
        \end{tabular}}
    \label{tab:diagnosishelpful}
    \vspace{-0.25in}
\end{table}

\vspace{-0.03in}
\boldparagraph{RQ}
Once the  \nNFRs{} were built,
the AEs met with the TSs and N-TSs involved in each case study to evaluate the \CRRI, in one session for each case study. 
When no issue was identified, AEs provided a proof of correctness. We focused on the stakeholders' ability to understand the feedback given by \tool to split the identified concerns into \emph{relevant} (TP) or \emph{spurious} (FP). Additionally, the stakeholders specified which part of the diagnosis helped them understand and perform this classification. The results are in Tbl.~\ref{tab:identifyIssues}.  

It took \tool between 21 and 49 seconds to analyze the \nNFR{} \CRRI in each case study, using a ThinkPad X1 Carbon with an Intel Core i7 1.80 GHz processor, 8 GB of RAM, and running 64-bit Ubuntu GNU/Linux 8. 
%
No vacuous conflicts were identified, likely
 due to the careful manual review of the rules prior to our analysis. To ensure correctness, we augmented the proof of absence of vacuous conflicts produced by \tool with
feedback produced by AutoCheck~\cite{feng-et-al-23}, an existing tool for analyzing vacuous conflicts. This evaluation, see Appendix~\ref{ap:eval}, confirmed the soundness of our approach.

\tool identified 20 situational conflicts, each involving 2-3 rules, and a conflicting situation. For each conflict, \tool also highlighted 2-3 out of 20-33 environment measures that are constrained by the conflicting situation to satisfy its conflicting condition. 
This means that the rules conflict only in highly limited environments.
Reviewing these conflicts without tool support is quite challenging, especially for case studies involving many rules (DAISY) or numerous environmental constraints  (e.g., DAISY, AutoCar, and ASPEN). 
Here, the diagnosis helped stakeholders  understand the issue by: (1) specifying which rules to focus on; (2) pinpointing the specific clauses causing the conflicts;  and (3) providing a trace specifying the situations in which the conflicts occur and the blocked response, thus aiding stakeholders in understanding the environments under which the conflicts occur.
\tool also identified eight redundancies involving 2-3 rules each.  The diagnosis provided all the necessary information to understand each redundancy, including identifying and extracting redundant rules and highlighting  redundant clauses. Stakeholders quickly confirmed that none of these redundancies were intentional, so
the identified redundancies were unnecessary and could have led to future inconsistencies.
Regarding restrictiveness, one case study, SafeSCADE, was found to have \nNFRs{} that were too restrictive, preventing the system from achieving two of its main purposes. 
The provided diagnosis directly helped the stakeholders.

Cases of insufficiency were identified in every case study, raising the total of 38 concerns. 
Determining the spuriousness of some concerns required a non-trivial  examination and consideration process.
Specifically, the stakeholders analyzed the traces provided by our diagnosis (e.g., Fig.~\ref{fig:unsatDiag}) and engaged in detailed discussions on whether their rules needed to address each identified concern, e.g., whether enabling patients to disconnect sensors tracking their health is a bad behaviour.
At the end, they identified spurious concerns only in four case studies.  
The identified spurious concerns were caused by their overly generic specifications. 
The feedback provided by \tool enabled the stakeholders to understand the meaning of the identified cases for all types of \CRRI, limiting spurious feedback to raised concerns. Thus, we successfully answered our RQ.

\boldparagraph{Discussion about diagnosis support for debugging {\CRRI}s}
Once the stakeholders understood and classified each identified \CRRI issue as non-spurious, we asked them, in the same session, to use the diagnosis to \emph{debug} the issue.  We recorded the resolution effort, i.e., the number of definitions/rules added, refined, merged or deleted, as a proxy to the resolution complexity. 
The results are in Tbl.~\ref{tab:diagnosishelpful}. 
The enumeration of redundant rules along with the highlighting of redundant clauses (e.g., Fig.~\ref{fig:unsatDiag}) provided by \tool enabled the stakeholders to easily resolve the cases of redundancy. They knew directly what to fix (mostly merge, delete, or refine).
\tool supported the stakeholders in resolving situational conflicts by:
(1) specifying the specific clauses causing the conflicts, which enabled stakeholders to focus on these rules, making the resolution process more targeted;
(2) providing the event and measure values representing the situations in which the conflicts occurred, which helped stakeholders consider the environmental constraints comprehensively while devising a resolution strategy (merging and refining both rules and definitions),
and (3) highlighting the blocked response in  the trace  (e.g., see Fig.~\ref{fig:unsatDiag}), which allowed stakeholders  study priority between the conflicting rule responses.

Regarding insufficiency, the stakeholders specifically appreciated the optimal size of the trace (e.g., Fig.~\ref{fig:unsatDiag}), as it helped them understand which events interfere with the identified concerns.
Regarding restrictiveness, the tool  provided all of the necessary information (i.e., pinpointed clauses on the rules restricting the system purpose) for the stakeholders to straightforwardly and trivially resolve the issue. Resolving restrictiveness issues (only two instances found across all the case studies) resulted in refining and merging the rules. 
In the case of AutoCar, checking restrictiveness helped refine constraints and correct incorrect responses, as well as establish priority between the rules, ultimately simplifying the rules.   
It also enabled the refinement of definitions and rules to be more specific (less restrictive).
%
For the majority of resolutions, stakeholders added a rule (60\%) or refined an existing one (23\%). They very rarely removed one (10\%), and even less frequently merged existing rules (7\%).
Our diagnosis supported the stakeholders  in resolving the \CRRI  when time allowed.  

\vspace{-0.03in}
\boldparagraph{Summary}
Our experiment demonstrated that our approach effectively and efficiently identifies \nNFR{} \CRRI. 
Both technical and non-technical stakeholders found the diagnosis to
be comprehensible and
helpful in resolving the identified issues. 
Notably, \tool allowed stakeholders to add missing rules and encouraged them to consider missing cases without constraining their imagination during the elicitation process. 
The inclusion of traces proved beneficial, enabling stakeholders to understand concerns that were not addressed by their \nNFRs{} and providing insights into causes of conflicts. 
However, stakeholders expressed the desire for a methodology to aid in issue resolution, as well as suggestions for potential resolution patches. They also would have preferred
using the tool during the elicitation process to help them identify and address potential issues proactively. Overall, the stakeholders really appreciated the assistance provided by \tool in addressing the \CRRI.

\vspace{-0.03in}
\boldparagraph{Threats to validity}
(1) 
The case studies were built by the co-authors of this paper. We mitigated this threat by separating the authors onto those who developed the validation approach and those who built the case studies, and ensured that a complete separation between them was maintained throughout the entire lifecycle of the project.
(2)
While we considered only nine case studies,  we  mitigated this treat by choosing them from different domains, at different development stages, and written by stakeholders with different types of expertise. Developing, collecting, evaluating these case studies took around
140 hours of work, resulting, to the best of our knowledge, in the most extensive collection of normative requirements to date.
%
(3)
The number of identified issues was relatively modest  because we 
started with existing systems where many issues had already been resolved, and followed the stakeholders' manual analysis.
Yet the stakeholders mentioned that they would have preferred to bypass this manual reviewing meeting and use our tool directly, as it would have saved them at least 10 hours. 
(4)~The earlier discussion of the effectiveness of \tool's diagnostic support was informed by the feedback from the stakeholders involved in the development of the requirements. While this feedback might be biased, it mirrors the intended real-world usage scenario of our approach, as \tool is designed to assist stakeholders with debugging the rules that they have elicited themselves, and thus their expertise is essential in this process. A more comprehensive evaluation involving a control study is left for future work. 
\vspace{-0.09in}
\section{Related Work}
\label{sec:relatedwork}

\vspace{-0.04in}
Below, we 
compare our work with existing 
approaches to identifying NFR and, specifically, \nNFRs{} well-formedness issues.


\vspace{-0.05in}
\boldparagraph{\upCRRI analysis of NFRs}
One of the primary challenges in verifying NFRs is resolving conflicts and trade-offs between multiple requirements engineering aspects~\cite{ChazetteSchneider2020,mandira-et-al-23,Lamsweerde-2009}. 
Formal techniques have paved the way for NFR conflict detection~\cite{Manna-Pnueli-92, Owre-Rushby-Shankar-95}. 
Researchers have also explored conflict and redundancy analysis using catalogs~\cite{Mairiza2011Zowghi}, or goal-oriented modelling~\cite{Lamsweerde-Darimont-Letier-98}~\cite{Matsumoto-et-al-17}. 
\st{In contrast to} \approach, \st{these other approaches necessitate technical expertise, making them less accessible to stakeholders without such skills.} 
\red{In contrast to these approaches, \approach has been developed specifically for helping non-technical stakeholders to elicit normative requirements.} 
Furthermore, these approaches are not designed to capture intricate non-monotonic conditions expressed through `if', `then', `unless', as well as time constraints, which are abundant in \nNFRs.



\vspace{-0.05in}
\boldparagraph{\upCRRI analysis of \nNFRs} 
To address \nNFR{} \CRRI, SLEECVAL analyzes conflicts and redundancies using a time variant of the CSP process algebra~\cite{Getir-Yaman-et-al-23,Getir-Yaman-et-al-23b}.
AutoCheck~\cite{feng-et-al-23-b} proposed using satisfiability checking and implemented a normative requirements' elicitation process by highlighting the causes for conflicts and redundancies. 
Compared to SLEECVAL,  \approach refines the definition of redundancy and conflicts to ensure the analysis results are complete while considering the effect of every rule in a given rule set. Compared to Autocheck, \approach can detect situational conflicts, insufficiency and restrictiveness.

%




\vspace{-0.07in}
\section{Conclusion}
\label{sec:conclusion}

\vspace{-0.03in}
We have introduced \approach, an approach to efficiently identify and debug \CRRI in normative requirements through the use of satisfiability checking.
The result of this  process is a diagnosis pinpointing the root causes of these issues, thus facilitating their resolution and supporting the normative requirements elicitation process. This approach enables stakeholders to analyze normative requirements on the fly, and use the diagnostic results to improve the quality of the requirements during the elicitation process.
The effectiveness and usability of \approach have been demonstrated via nine case studies with a total of 233 \nNFRs{} elicited
by  stakeholders from diverse backgrounds.
This work is part of a longer-term goal of supporting the development of systems that operate ethically, responsibly, and in harmony with human values and societal norms.
In the future, we plan to
 further develop the debugging support by
  (semi-)automatically generating and suggesting patches to \CRRI.
As requested by multiple stakeholders, we also plan to design a systematic process for analyzing the \CRRI while eliciting the \nNFRs. 
Finally, we have focused on well-formedness constraints but envision the need to analyze and debug a broader range of properties such as the scenario-based ones~\cite{Gregoriades-et-al-05}.

\section*{Acknowledgements}

This work has received funding from the UKRI project EP/V026747/1 `Trustworthy Autonomous Systems Node in Resilience', the UKRI projects EP/V026801/2 and EP/R025479/1, the Assuring Autonomy International Programme, \red{NSERC and the Amazon Research Award}, the Royal Academy of Engineering Grant No CiET1718/45.

\bibliographystyle{splncs04}
\bibliography{refs}


\newpage
\appendix
\section*{Appendix}

Sec.\ref{ap:norm} presents \dsl normalization. Sec.\ref{ap:extendedsem} provides the \emph{bounded semantics} of normalized \dsl. Sec.\ref{ap:translate} presents the translation from \dsl to \fol.
Sec.\ref{app:sanitizationSoundness} provides the proof for the correctness of \tool in checking \emph{vacuous conflict} and \emph{redundancy}.
Sec.\ref{ap:translateSituation} formally describes the \fol encoding of \emph{situational conflicts}.
Sec.\ref{ap:treereduction} provides details on the derivation tree and the reduction process for causal unsatisfiability proofs.
Sec.\ref{ap:iae} defines the active conditions for an atomic element in clausal proofs.
Sec.\ref{ap:eval} provides additional evaluation results.

\section{Rule Normalization} \label{ap:norm}


\begin{figure*}
    $\normalize(\textit{resp}, \prop) = \begin{cases}
     \{\prop \Rightarrow \within{\event
     }{\term}\} & \textit{if } \textit{resp} = \within{\event}{\term} \\
     \{\prop \Rightarrow \within{\notS \; \event
     }{\term}\} & \textit{if } \textit{resp} = \within{\notS \; \event}{\term} \\
     \normalize(\textit{resp}_1 , \prop \; \andS \; \notS \; \prop')  \cup \normalize(\textit{resp}_2, \prop \; \andS \; \prop')  & \textit{if } \textit{resp} = \textit{resp}_1 \textbf{ unless } \prop' \textbf{ then } \textit{resp}_2\\
     
          \normalize(\textit{resp}_1 , \prop \; \andS \; \notS \; \prop')   & \textit{if } \textit{resp} = \textit{resp}_1  \textbf{ unless } \prop'\\
          
     \{ \normalize(\within{\event}{\term}, p) \textbf{ otherwise } \dobg \mid  \dobg \in \normalize( \textit{resp}_2, \top)\}
     & \textit{if } \srule_{op} = \within{\event}{\term} \textbf{ otherwise } \textit{resp}_2  \\
\end{cases}$ 

    \caption{Function $\normalize$ takes \textit{resp} and $\prop$ and returns a set of normalized pseudo-rules.  }
    \label{fig:normalize}
\end{figure*}

\begin{figure*}
    \scalebox{0.75}{
        \begin{tabular}{c}
              Original SLEEC Rule  \\
              $r_{o} = \rulesyntax{\event_1 \; \andS \; \prop_1}{\within
        {\event_2}{\term_1} \textbf{ otherwise } (\within{\event_3}{ \term_2} \textbf{ unless } \prop_3 \textbf{ then } \within{\event_4}{ \term_3})}$ \\
        \hline
        Normalized SLEEC Rules\\
        $r_{n1} =\rulesyntax{\event_1}{(\prop_1 \Rightarrow (\within{\event_2}{\term_1})) \textbf{ otherwise } (\notS \; \prop_3 \Rightarrow \within{\event_3}{ \term_2})
        }$\\
        $r_{n2} = \rulesyntax{\event_1}{(\prop_1 \Rightarrow (\within{\event_2}{\term_1})) \textbf{ otherwise } (\prop_3 \Rightarrow \within{\event_4}{ \term_3})
        }$
         \end{tabular}
     }
     \caption{An example of SLEEC Rule normalization. Given an original SLEEC rule $r_{o}$, applying function $\normalize$ yields two normalized rules $r_{n1}$ and $r_{n2}$.}\label{tab:normalizationExample}
\end{figure*}

In this section, we present the normalization function, \normalize, which converts an original SLEEC rule to a set of normalized \dsl rules. The original \dsl rule follows the syntax: $\rulesyntax{\event \wedge \prop}{\textit{resp}}$ where $\event$ is an event symbol, $\prop$ is a proposition and $\textit{resp}$ is a response. A response is one of the following:
\begin{enumerate}
    \item $\within{\notS \; \event}{\term}$
    \item $\within{\event}{\term}$
    \item $\within{\event}{\term} \textbf{ otherwise } \textit{resp}$
    \item $\textit{resp}_1 \textbf{ unless } \prop textbf{ (\textit{ then } }\textit{resp}_2)?$  where the expression $(*)?$ indicates $*$ is optional. 
\end{enumerate}

Let an original \dsl rule \quoted{$\srule_{o} = \rulesyntax{(\event \wedge \prop)}{\textit{response}}$} be given. The result of normalizing $\srule_{o}$ is the set of normalized rules = $\{ \rulesyntax{\event}{\dobg} \mid \dobg \in \normalize(\textit{resp}, \prop)\}$ where the normalization function \normalize is defined in Fig.~\ref{fig:normalize}. Given a response $\textit{resp}$, $\normalize$ flattens $\textit{resp}$ into a set of obligation chains by traversing the nested structure of $\textit{resp}$ top-down, and then merges the normalization results bottom-up. Note that in a nested response, a chain of $\textit{unless}$ is left associative (e.g., $A$ \textit{unless} $B$ \textit{unless} $C$ is equivalent as (($A$ \textit{unless} $B$) \textit{unless} $C$) ), and $\textit{otherwise}$ has a higher precedence than $\textit{unless}$ by default (e.g., $A$ \textit{unless} $B$ \textit{otherwise} $C$ is equivalent to $A$ \textit{unless} ($B$ \textit{otherwise} $C$ ). $\normalize$ also records and recursively distributes the triggering condition $\prop$ to each case. The set of obligation chains returned by $\normalize$ can then be turned into a set of normalized rules by disturbing the triggering event $\event$ to them.

\begin{corollary}
For any original SLEEC rule with $n$ syntax tokens, the size of the normalized SLEEC rule is $O(n)$ 
\end{corollary}

\begin{example}
    Consider the original SLEEC rule $r_{o}$ shown in Fig.~\ref{tab:normalizationExample}. Applying the normalization function $\normalize$ on $r_{o}$ yields two normalized rules $r_{n1}$ and $r_{n2}$ shown in Fig.~\ref{tab:normalizationExample}.
\end{example}

The semantics of the normalized \dsl is shown in Fig.~\ref{fig:semantics}.

\begin{figure*}
    \centering 
    \scalebox{0.8}{
        \begin{tabular}{lll}
        $\fos \models_i \prop$  & iff & $\measureassign{i}(\prop) $\\
        
        $\fos \models_{i} \within{\event}{\term}$  & iff & $\exists j \in [i, n] . (\event \in \eventocc{j} \wedge \timestamp{j} \in [\timestamp{i}, \timestamp{i} +\measureassign{j}(\term)] )$ \\
    
        $\fos \not\models_{i}^{j} \within{\event}{\term}$ & iff & $\timestamp{j} = \timestamp{i} + \measureassign{i}(\term) \wedge \forall j' \in [i, j](\event \not\in \eventocc{j'})$ \\
    
        $\fos \models_{i} \within{\notS \; \event}{\term}$ & iff & $\exists j (\fos \not\models_{i}^{j} \within{\event}{\term})$ \\
    
        $\fos \not\models_{i}^{j} \within{\notS \; \event}{\term}$  & iff & $\fos \models \within{\event}{\term} \wedge \forall j'\in [i,j) (\fos \not\models_{i}^{j} \within{\event}{\term} )  $\\
        
        $\fos \models_i (\prop \Rightarrow \obg)$ & iff & $\fos \models_{i} \prop \Rightarrow \fos \models_{i} \obg$\\
    
         $\fos \not\models_{i}^{j} (\prop \Rightarrow \obg)$ & iff & $\fos \models_{i} \notS \; \prop \wedge \fos \not\models_{i}^{j} \obg$\\
    
         $\fos \models_{i} \cobg^+ \otherwise \dobg$ & iff & $\fos \models_{i} \cobg^+ \vee \exists j (\fos \not\models_{i}^{j} \cobg^+ \wedge \fos \models_{j} \dobg)$ \\
    
         $\fos \not\models_{i}^{j} \cobg^+ \otherwise \dobg$ & iff & $\exists j' \in [i, j] (\fos \not\models_{i}^{j'} \cobg^+ \wedge \fos \not\models_{j'}^{j} \dobg)$\\
    
         $\fos \models \rulesyntax{\event \sand \prop}{\dobg}$ & iff & $\forall i\in [1,n] ((\event \in \eventocc{i} \wedge \measureassign{i}(\prop)) \Rightarrow \fos \models_{i} \dobg)$\\
         
         \hline
         $\fos \models_{i} \notS \dobg$ & iff & $\exists j (\fos \not\models_{i}^{j} \dobg)$\\
    
         $\fos \models \factsyntax{\event \sand \prop}{\dobg}$ & iff & $\exists i\in [1,n] (\event \in \eventocc{i} \wedge \measureassign{i}(\prop) \wedge \fos \models_{i} \dobg)$\\
    
         $\fos \models \factsyntax{\event \sand \prop}{\notS \dobg}$ & iff & $\exists i\in [1,n] (\event \in \eventocc{i} \wedge \measureassign{i}(\prop) \wedge \fos \models_{i} \notS \; \dobg)$\\
        
        \end{tabular}
    }

 \caption{\small Semantics of normalized $\dsl$ defined over trace $\fos = (\eventocc{1}, \measureassign{1}, \timestamp{1}) \ldots (\eventocc{n}, \measureassign{n}, \timestamp{n})$.}
  \label{fig:semantics}
\end{figure*}

     


\section{SLEEC DSL Semantic Extensions}
\label{ap:extendedsem}

The \emph{bounded semantics} of \dsl is presented in in Fig.~\ref{fig:boundsemantics}. The bounded semantics are defines over traces to describe the satisfaction of \dsl rule up to a time point $k$.

\begin{figure}
    \centering 
    \scalebox{0.9}{
        \begin{tabular}{lll}        
        $\fos \vdash_{i}^{k} \obg$  & iff & $\neg (\exists j \in [i, k]. \fos \not\models_{i}^{k} \obg))$ \\
        
        $\fos \vdash_i^k (\prop \Rightarrow \obg)$ & iff & $\fos \models_{i} \prop \Rightarrow \fos \vdash_{i}^{k} \obg$\\
    
         $\fos \vdash_{i}^{k} \cobg^+ \otherwise \dobg$ & iff & \makecell{$\fos \vdash_{i} \cobg^+ \vee \exists j (j \le k \wedge \fos \not\models_{i}^{j} \cobg^+ $ \\ $\wedge \fos \vdash_{j}^{k} \dobg)$} \\

         $\fos \vdash^k \rulesyntax{\event \sand \prop}{\dobg}$ & iff & \makecell{$\forall i\in [1,k] ((\event \in \eventocc{i} \wedge \measureassign{i}(\prop))$ \\ $ \Rightarrow \fos \vdash_{i}^k \dobg)$}\\
        \end{tabular}
    }
 \caption{\small Bounded Semantics of normalized $\dsl$ up to a time point $k$ defined over trace $\fos = (\eventocc{1}, \measureassign{1}, \timestamp{1}) \ldots (\eventocc{n}, \measureassign{n}, \timestamp{n})$ where $n \ge k$.}
 \label{fig:boundsemantics}
\end{figure}

\section{Translation of SLEEC DSL to \fol}~\label{ap:translate}
In this section, we provide the translation function $T$  in Tbl.~\ref{tab:translate}, and prove the correctness claim of the translation prove the correctness of the translation in Thm.~\ref{thm:transcorrectness}.

\begin{figure*}
    \centering
    \scalebox{0.8}{
    \begin{tabular}{lll}            
                $T(\rulesyntax{\event \; \andS \; \prop}{\dobg})$ & \multicolumn{2}{l}{$\rightarrow$ $\forall \obj{\event}: \Class{\event}\exists \obj{\measures}:\Class{\measures}(\obj{\measure}.\textit{time} = \obj{\event}.\textit{time} \wedge (T^{*}(\prop, \obj{\measures}) \Rightarrow T^{*}(\dobg, \obj{\measures})))$}  \\
                
                $T(\factsyntax{\event \; \andS \; \prop}{\dobg})$ & \multicolumn{2}{l}{ $\rightarrow$ $\exists \obj{\event}: \Class{\event}\exists \obj{\measures}:\Class{\measures}(\obj{\measure}.\textit{time} = \obj{\event}.\textit{time} \wedge (T^{*}(\prop, \obj{\measures}) \wedge T^{*}(\dobg, \obj{\measures})))$} \\
                
                $T^{*}(\cobg_1 \otherwise \cobg_2 \ldots \cobg_n, \; \obj{\measures})$& \multicolumn{2}{l}{$\rightarrow$ $T^{*}(\cobg_1, \obj{\measures}) \vee \exists \obj{\measures}_j:\Class{\measures}(\textbf{Violation}(\cobg_1, \obj{\measures}, \obj{\measures}_j) \wedge T^{*}(\cobg_2 \otherwise \ldots \cobg_n, \obj{\measures}_j))$} \\

                $T^{*}(\notS \; \dobg, \; \obj{\measures})$& \multicolumn{2}{l}{$\rightarrow$ $\neg T^{*}(\dobg)$} \\

                $T^{*}(\prop \Rightarrow \obg,  \obj{\measures})$&  \multicolumn{2}{l}{$\rightarrow$ $\neg T^{*}(\prop, \obj{\measures}) \wedge T^{*}(\obg, \obj{\measures})$} \\

                $\textbf{Violation}(\prop \Rightarrow \obg,  \obj{\measures}, \obj{\measures}_j)$&  \multicolumn{2}{l}{$\rightarrow$ $T^{*}(\prop, \obj{\measures}) \wedge \textbf{Violation}(\obg, \obj{\measures}, \obj{\measures}_j)$} \\
                
                $T^{*}(\within{\event}{\term},  \obj{\measures})$&  \multicolumn{2}{l}{$\rightarrow$ $\exists \obj{\event}:\Class{\event}(\obj{\measures}.\textit{time} \le \obj{\event}.\textit{time} \le \obj{\measures}.\textit{time} + T^{*}(\term, \obj{\measures}))$} \\

                $\textbf{violation}(\within{\event}{\term},  \obj{\measures}, \obj{\measures}_j)$&  \multicolumn{2}{l}{$\rightarrow$ $ \obj{\measures}_j.\textit{time} = \obj{\measures}.\textit{time} + T^{*}(\term, \obj{\measures})  \wedge \neg T^{*}(\within{\notS \; \event}{\term},  \obj{\measures})$ } \\

                $T^{*}(\within{\notS \; \event}{\term},  \obj{\measures})$& \multicolumn{2}{l}{ $\rightarrow$ $\neg T^{*}(\within{\notS \; \event}{\term},  \obj{\measures})$} \\

                $\textbf{Violation}(\within{\neg \event}{\term},  \obj{\measures}, \obj{\measures}_j)$&  \multicolumn{2}{l}{$\rightarrow$ \makecell{$ 
                (\obj{\measures}.\textit{time} \le \obj{\measures}_j.\textit{time} \le \obj{\measures}.\textit{time} + T^{*}(\term, \obj{\measures})) \wedge \exists \obj{\event}:\Class{\event}(\obj{\measures}_j.\textit{time} = \obj{\event}.\textit{time})$  \\ $\neg (\exists \obj{\event}_1:\Class{\event} (\obj{\measures}.\textit{time} \le \obj{\event}_1.\textit{time} < \obj{\measures}_j.\textit{time}) )$}} \\
                
                $T^{*}(c, \obj{\measures})$ $\rightarrow$ $c$  &  $T^{*}(\measure, \obj{\measures})$ $\rightarrow$ $\obj{\measures}.\measure$ &  \\
                $T^{*}(-\term, \obj{\measures})$ $\rightarrow$ $-T^{*}(\term, \obj{\measures})$& $T^{*}(\term_1 + \term_2, \obj{\measures})$ $\rightarrow$ $T^{*}(\term_1, \obj{\measures}) + T^{*}(\term_2, \obj{\measures})$ & $T^{*}(c \times \term, \obj{\measures})$ $\rightarrow$ $T^{*}(c, \obj{\measures}) \times T^{*}(\term, \obj{\measures})$\\
                
                $T^{*}(\top, \obj{\measures})$ $\rightarrow$  $\top$ & $T^{*}(\term_1 > \term_2, \obj{\measures})$ $\rightarrow$  $T^{*}( \term_1, \obj{\measures}) > T^{*}( \term_2, \obj{\measures})$ & $T^{*}(\term_1 = \term_2, \obj{\measures})$ $\rightarrow$  $T^{*}( \term_1, \obj{\measures}) = T^{*}( \term_2, \obj{\measures})$\\
                $T^{*}(\neg \prop, \obj{\measures})$ $\rightarrow$  $\neg T^{*}( \prop, \obj{\measures})$& $T^{*}( \prop_1 \wedge \prop_2, \obj{\measures})$ $\rightarrow$   $T^{*}( \prop_1, \obj{\measures}) \wedge T^{*}( \prop_2, \obj{\measures}) $ & $T^{*}( \prop_1 \vee \prop_2, \obj{\measures})$ $\rightarrow$   $T^{*}( \prop_1, \obj{\measures}) \vee T^{*}( \prop_2, \obj{\measures}) $
                \end{tabular}  
    }
\caption{Translation rules from \dsl to \fol. Given a \dsl rule $\srule$, $T$ translates $\srule$ to an \fol formula using the translation function $T^{*}$. The function $T^{*}$ recursively visits the elements (i.e., term, proposition and obligations) 
of $\srule$ and translates them into \fol constraints under a relational object $\obj{\measures}$ representing the measures when $\srule$ is triggered.} 
\label{tab:translate}
\end{figure*}

We now state and prove the correctness of the \fol encoding:

\begin{theorem}[Correctness of the \fol translation]\label{thm:transcorrectness}
Let a set of rules $\ruleset$ and facts $\facts$ in \dsl be given. There exists a trace $\fos = (\eventocc{1}, \measureassign{1}, \timestamp{1}),$ $ (\eventocc{2}, \measureassign{2}, \timestamp{2}),$ $\ldots  (\eventocc{n}, \measureassign{n}, \timestamp{n})$ such that $\fos \in \lang{\ruleset} \cap \lang{\facts}$ if and only if $\translate{\ruleset} \wedge \translate{\facts} \wedge \textit{axiom}_{mc}$ has a satisfying solution $(\domain, v)$.
\end{theorem}

\begin{sop}
    We prove the forward direction by constructing a satisfying solution $(\domain, v)$ to $\translate{\ruleset} \wedge \translate{\facts}$ from the trace $\fos \in \lang{\ruleset} \cap \lang{\facts}$. For every state $(\eventocc{i}, \measureassign{i}, \timestamp{i}) \in \fos$, we follow the construction rules: (1) for every event $\event \in \eventocc{i}$, add a relational object $\obj{e}$ of class $\Class{e}$ such that $v(\obj{e}.\textit{ext}) = \top$ and  $v(\obj{e}.\textit{time}) = \timestamp{i}$; and (2) add a relational object $\obj{\measures}$ such that $\obj{\measures}.\textit{time} = \timestamp{i}$ and $v(\obj{\measures}.m) = \measureassign{i}(\measure)$ for every measure $m \in \measures$. We then prove that the constructed $(\domain, v)$ is a solution to   $\translate{\ruleset} \wedge \translate{\facts} \wedge \textit{axiom}_{mc}$. 

    We prove the backward direction by constructing  $\fos$ from a satisfying solution $(\domain, v)$ to $\translate{\ruleset} \wedge \translate{\facts} \wedge \textit{axiom}_{mc}$. The construction maps every relational object $\obj{\event}$ and $\obj{\measures}$ to some state $(\eventocc{i}, \measureassign{i}, \timestamp{i}) \in \fos$, where (1) $\event \in \eventocc{i}$ if $v(\obj{\event}).\textit{ext} = \top \wedge v(\obj{\event}).\textit{time} = \timestamp{i}$; and (2) $\measureassign{\measure} = v(\obj{\measures}.\measure)$ for every $\measure \in \measures$ if $v(\obj{\measures}).\textit{ext} = \top \wedge v(\obj{\measures}).\textit{time} = \timestamp{i}$. We then prove $\fos \in \lang{\ruleset} \cap \lang{\facts}$ 
    and conclude the proof. 
\end{sop}

\section{Soundness of the \approach approach}\label{app:sanitizationSoundness}
\begin{theorem}[Vacuous Conflict as Satisfiability]\label{thm:vcdef}
    Let a set of rules $\ruleset$ be given. A rule $\srule \in \ruleset$ is vacuously conflicting if and only if $\translate{\ruleset} \wedge \translate{\triggerfunc{\srule}} \wedge \textit{axiom}_{mc}$ is UNSAT.
\end{theorem}

\begin{proof}
    By Lemma~\ref{lemma:vcdef}, $\srule$ is vacuously conflicting if and only if $\lang{\ruleset} \cap \lang{\triggerfunc{\srule}}$ is empty. The emptiness checking problem can be encoded as the satisfiability problem by Thm.~\ref{thm:transcorrectness}.
\end{proof}

Given a rule $\srule \in \ruleset$, lemma.~\ref{lemma:redundantdef} enables detection of redundancy by checking the emptiness of $\lang{\ruleset \setminus \srule} \cap \lang{\noncompliant{\srule}}$. 

The emptiness checking encoded as a \fol satisfiability problem:

\begin{theorem}[Redundancy via Satisfiability]\label{Thm:redviasat}
    Let a set of \dsl rules ($\ruleset$) be given. A rule $\srule \in \ruleset$ is redundant if and only if $\translate{\ruleset \setminus \srule} \wedge \translate{\noncompliant{\srule}} \wedge \textit{axiom}_{mc}$ is UNSAT. 
\end{theorem}

\begin{proof}
    By definition of redundancy (Def.~\ref{def:redudent}), $\srule$ is redundant if and only if there does not exist a $\fos$ such that $\fos \in \lang{\ruleset \setminus \srule}$ and $\fos \not\models \srule$. By Lemma.~\ref{lemma:concern}, $\srule$ is redundant if and only if there does not exist a $\fos$ such that $\fos \in \lang{\ruleset \setminus \srule} \cap \lang{\noncompliant{\srule}}$. By Thm.~\ref{thm:transcorrectness}, $\srule$ is redundant if and only if $\translate{\ruleset \setminus \srule} \wedge \translate{\noncompliant{\srule}} \wedge \textit{axiom}_{mc}$ is UNSAT .
\end{proof}


\section{\fol Encoding for Situational Conflicts}\label{ap:translateSituation}

In this section, we first define the states of obligation in Def.~\ref{def:triggered}, and then use it to prove sufficient condition for situational conflict in Thm.~\ref{lemma:sc}. Next, we present the \fol encoding for the sufficient condition of situational conflict in Tab.~\ref{tab:sctranslate}, and finally provide a sketch of the proof the correctness of the encoding (Thm.~\ref{thm:condconfencoding}).

\begin{definition}[State of Obligations]\label{def:triggered}
Let a rule set $\ruleset$, a rule $\srule \in \ruleset$ and an $\srule$-triggering situation ($\fos^k_0$, $\Vec{\measures_{k}}$) be given. The time point $k$ is the last time point of $\fos^k_0$, and it is also when $\srule$ is triggered. The status of rules, obligation chains, conditional obligations and obligations are defined as follows:
\begin{itemize}[label={},leftmargin=0cm]
    \item \textbf{Triggered}: A rule $\srule = \rulesyntax{\event \wedge \prop}{\dobg}$ is \textit{triggered} at time point $i$ if $i \le k$ and $\fos^k_0 \models_i e$ and $\fos^k_o \models \prop$. If a rule $\srule = \rulesyntax{\event \wedge \prop}{\dobg}$ is triggered at $i$, then the obligation chain $\dobg$ is \textit{triggered} at $i$. If an obligation chain $\dobg = \cobg \otherwise \dobg'$ is triggered at $i$, then (1) the \emph{conditional obligation} $\cobg$ is triggered at $i$ and (2) $\dobg'$ is \textit{triggered} at $j > i$ if $\cobg$ is \textit{violated} at $j$ or $\cobg$ is \textit{blocked} at $i$. If a conditional obligation $\prop \Rightarrow \obg$ is triggered at $i$ and $\prop$ is evaluated to $\top$ at $i$, then $\obg$ is triggered at $i$.  
    
    \item \textbf{Fulfilled}: An obligation $\obg$ is \textit{fulfilled} at time point $j \le k$ if it is triggered at some time point $i \le j$ and
    $\fos_0^j \models_{i} \obg$. A conditional obligation $\prop \Rightarrow \obg$ (triggered at $i$)
    is \textit{fulfilled} at $j$ if its obligation $\obg$ is fulfilled at $j$ or $\prop$ evaluated to $\bot$ at $i$. An obligation chain $\dobg = \cobg \otherwise \dobg'$ 
    is fulfilled at $j$ if $\cobg$ is fulfilled at $j$ or $\dobg'$ is fulfilled at $j$. 
    
    \item \textbf{Violated}: An obligation $\obg$ (\textit{triggered} at $i$) is \textit{violated} at a time point $j \le k$ if $\fos^k_0 \not\models_{i}^{j} \obg$. A conditional obligation $\prop \Rightarrow \obg$ is \textit{violated} at time point $j$ if $\obg$ is \textit{violated} at $j$. An obligation chain $\dobg = \dobg' \otherwise \cobg$ is \textit{violated} at time point $j$ if $\dobg'$ is \textit{violated} at some point $j' \le j$, $\cobg$ is \textit{triggered} at $j$  and $\cobg$ is \textit{violated} at $j'$.
    
    \item \textbf{Active}: An obligation, conditional obligation and obligation chain are \textit{active} at time point $j$ if they are \textit{triggered} at some time point $i \le j$ and are not \textit{fulfilled} and \textit{violated} at any time point $j' \in [i,j]$.
    
    \item \textbf{Forced}: An obligation chain $\dobg$ is \textit{forced} at time point $j \ge k$ if $\dobg$ is \textit{active} at $j$. If an obligation chain $\dobg = \cobg^+ \otherwise \dobg'$ is \textit{forced}, and $\dobg'$ is \textit{blocked} at the time point $j'$ when $\cobg^+$ expires (s.t., $\timestamp{j'} = \timestamp{j} + \measureassign{j}$), then $\cobg^+$ is \textit{forced} at time point $j$. If a conditional obligation $\prop \Rightarrow \obg$ is \textit{forced} at $j$, and $\obg$ is triggered at $j$, then $\obg$ is \textit{forced} at $j$.
    
    \item \textbf{Blocked}: An obligation $\obg$ (triggered at $i$) is \textit{blocked} at time point $j$ if it is \textit{active} and there is an obligation $\obg'$ such that (1) $\obg'$ is \textit{forced} at time point $j$; (2) if $\obg = \within{\event}{\term} $ and $\obg' = \within{\neg \event}{\term'}$ then $\measureassign{i}(\term) + \timestamp{i} \le \measureassign{i'} (\term') + \timestamp{i'}$; and (3) if $\obg = \within{\neg \event}{\term} $ and $\obg' = \within{ \event}{\term'}$ then $\measureassign{i}(\term) + \timestamp{i} \ge \measureassign{i'} (\term') + \timestamp{i'}$. A conditional obligation $\prop \Rightarrow \obg$ is \textit{blocked} at $j$ if $\obg$ is \textit{blocked} $j$ and $\prop$ evaluates to $\top$ at $j$. An obligation chain $\cobg^+ \otherwise \dobg$ is \textit{blocked} at time point $j$ if $\cobg^+$ is \textit{blocked} at some time point $j$ and $\dobg$ is \textit{blocked} at time point $j'$ when $\cobg$ expires (s.t., $\timestamp{j'} = \timestamp{j} + \measureassign{j}$).
    
    Noticed that there are circular dependencies between forced and blocked obligation chains. Fortunately, by encoding the status definitions into \fol, we can leverage \fol solver's ability to incrementally and lazily unroll the necessary definitions to resolve the dependencies.
    
\end{itemize}
\end{definition}



We now provide the Proof of Lemma~\ref{lemma:sc}.
\begin{sop}
    Proof by contradiction, we assume there exists a trace $\fos$ such that $\fos$ is an extension to $\fos_0^k$ and is also consistent with $\Vec{\measures_k}$. 
    Since $\fos \in \lang{\ruleset}$, then $\fos$ fulfill the obligation chain $\dobg = (\cobg_1, \ldots \cobg_n)$ triggered at $k$ ($\fos \models_k \dobg$). Therefore, by the semantics of obligation chain fulfillment, either $\fos \models_k \cobg_1$ or $\cobg_1$ is positive and there exists a time point $k' \ge k$ such that $\fos \not\models_{k}^{k'} \cobg_1$ and $\fos \models_{k'} (\cobg_2, \ldots \cobg_n)$

    Since $\dobg$ is blocked at $k$, by Def.~\ref{def:triggered}, the obligation $\obg_1$ in $\cobg_1$ is blocked at $k$, and for every conditional obligation $\cobg_m$,  the obligation $\obg_m$ is $\cobg_m$ is blocked at the unique time when $\obg_{m-1}$ is violated (the violation time is unique since $\cobg_1 \ldots \cobg_{n-1}$ are all positive). Therefore, it is sufficient to show that if an obligation $\obg$ is blocked, then $\fos$ does not fulfill $\obg_{m}$. There are two cases: 

    First, we consider the case $\obg = \within{\event}{\term}$. There is an event occurred at some time point $k \ge j$ where $\timestamp{j} \le \timestamp{k} \le \timestamp{j} + \measureassign{j}$. Since $\obg$ is blocked, then there exists an obligation forced by some rule $
\srule'$ (triggered at $j'$)
    such that $\obg' = \within{\neg \event}{\term'}$ where $\measureassign{j'}(\term') \ge \measureassign{j}$. Therefore, the occurrence of $\event$ at time point $k$ violates $\srule'$. Contradiction.

    The case where $\obg = \within{\neg \event}{\term}$ can be proved analogously. \qed
\end{sop}

\begin{remark}\label{remark:tanssc}
The sufficient condition for situational conflict defined in Lemma~\ref{lemma:sc} is not a necessary condition, as some situational conflicts do not require a rule's response to be blocked at the last state of the situation. Let's consider the set of rules $\{\srule_1, \srule_2, \srule_3, \srule_4\}$, where $\srule_{1} = \rulesyntax{\event_1}{\within{\event_2}{5}}$, $\srule_{2} = \rulesyntax{\event_3}{\within{\neg \event_2}{4}}$, $\srule_{3} = \rulesyntax{\event_4}{\within{\event_3}{3}}$, and $\srule_{4} = \rulesyntax{\event_1}{\within{\neg \event_3}{1}}$. The rule $\srule_1$ is situational conflicting in the situation ($\fos^1$, *) where $\fos^1 = ({\srule_1, \srule_2, \srule_3, \srule_4}, \measureassign{1}, 1)$ because, according to $\srule_1$ and $\srule_2$, $\event_{2}$ must occur within the interval $(4, 5]$. Additionally, based on $\srule_3$ and $\srule_4$, the event $\event_{3}$ must occur at a time $t \in (1, 3]$. For all possible values of $t$, according to $\srule_{2}$, $\event_{2}$ cannot occur within the interval $(t, t+4]$, which covers the interval $(4, 5]$ and thus conflicts with $\srule_1$. However, in the situation $\fos^1$, the obligation of $\srule_1$ is not blocked at time point $1$. We refer to the situational conflicts caused by forced obligations ``after'' the situation as``transitive situation conflicts''. Identifying situations that lead to transitive situation conflicts is not easily expressed as satisfiability (e.g., we need to cover the entire range of $t$ in the example) and is left as future work.
\end{remark}

We present the \fol encoding in Fig.\ref{tab:sctranslate} to describe the situation for a rule to be situational conflicting. Given a set of rules $\ruleset$ and a rule $\srule \in \ruleset$, every satisfying solution to the \fol formula $\sctranslate{\ruleset, \srule}$ represents a situation where $\srule$ is situational conflicting. The top-level encoding $\sctranslate{\srule, \ruleset}$ is presented in part (1) of Tab.\ref{tab:sctranslate}, which describes the existence of a situation ($\fos^k_0$, $\Vec{\measures_k}$) where $\srule$ is triggered at the last state of $\fos^k_0$ ($\obj{\measures}$). The situation should be non-violating ($\timedtranslate{\ruleset, \obj{\measures}.\textit{time}}$) and should block the response of $\srule$ ($\blocked{\dobg} {\obj{\measures}}{\obj{\measures}}$). The \fol encoding for non-violating and obligation blocking is presented in parts (2) and (3) of Tab.~\ref{tab:sctranslate}, respectively.

Note that eagerly expanding the definition of obligation blocking blows up the size of encoding exponentially (with respect to the number of obligations in $\ruleset$) due to the transitive dependencies between blocked obligations and forced obligations. To avoid the blow-up, we lazily expand the definition of obligation blocking by introducing an \textit{internal} class of relational object $\Class{\textit{block}{\obg}}$ for every obligation $\obg$ in $\ruleset$ to indicate if and when $\obg$ is blocked. The axiom $\textit{axiomBlock}{\obg}$ is added to describe the definition of a blocked obligation (Def.~\ref{def:triggered}) and is lazily applied to relational objects of $\Class{\textit{block}_{\obg}}$ in a given domain.

\begin{figure*}
    \scalebox{0.7}{
    \begin{tabular}{lcl}
                \centering   $\sctranslate{\rulesyntax{\event \wedge \prop}{\dobg}, \ruleset}$ & \multicolumn{2}{c}{\makecell{$\rightarrow \exists \obj{\event}:\Class{\event}\exists \obj{\measures}:\Class{\measures}(\obj{\measures}.\textit{time} = \obj{\event}.\textit{time} \wedge T^*(\prop, \obj{\measures}) \wedge \blocked{\dobg} {\obj{\measures}}{\obj{\measures}}$ \\
                $ \wedge \timedtranslate{\ruleset, \obj{\measures}.\textit{time}}) \wedge \textit{axiom}_{mc} \wedge \textit{axiomBlock}_{\obg}$ for every obligations $\obg$ in  $\ruleset$ }}  \\
                \hline
                \hline
                $\timedtranslate{\rulesyntax{\event \wedge \prop}{\dobg}, \textit{end\_time}}$ & \multicolumn{2}{c}{\makecell{$\rightarrow$ $\forall \obj{\event}: \Class{\event}(\obj{\event}.\textit{time} \le end\_time \Rightarrow \exists \obj{\measures}:\Class{\measures}$ \\$(\obj{\measure}.\textit{time} = \obj{\event}.\textit{time} \wedge (T^{*}(\prop, \obj{\measures}) \Rightarrow \timedtranslatestar{\dobg, \obj{\measures}, \textit{end\_time}})))$}}  \\
                \hline
                
                $\timedtranslatestar{\obg_1 \otherwise \ldots \obg_n, \obj{\measures}, \textit{end\_time}}$& \multicolumn{2}{c}{$\rightarrow$ $\timedtranslatestar{\obg_1, \obj{\measures}, \textit{end\_time}} \vee \exists \obj{\measures}_v:\Class{\measures} (\textbf{violation}(\obg_1, \obj{\measures}, \obj{\measures}_v) \vee 
    \timedtranslatestar{\obg_n, \obj{\measures}_v, \textit{end\_time}})$} \\
    \hline
                $\timedtranslatestar{\within{\event}{\term},  \obj{\measures}, \textit{end\_time}}$&
                \multicolumn{2}{c}{$\rightarrow$ $\exists \obj{\event}:\Class{\event}(\obj{\measures}.\textit{time} \le \obj{\event}.\textit{time} \le \obj{\measures}.\textit{time} + T^{*}(\term, \obj{\measures}) \; \vee \; \obj{\measures}.\textit{time} + T^{*}(\term, \obj{\measures}) > \textit{end\_time} $}\\
                \hline

                $\timedtranslatestar{\within{\neg \event}{\term},  \obj{\measures}, \textit{end\_time}}$&
                \multicolumn{2}{c}{$\rightarrow$ $\neg (\exists \obj{\event}:\Class{\event}(\obj{\measures}.\textit{time} \le \obj{\event}.\textit{time} \le \textsc{Min}(\obj{\measures}.\textit{time} + T^{*}(\term, \obj{\measures}), ,\textit{end\_time})) $}\\
                
                \hline\hline

                $\blocked{(\prop \Rightarrow \within{\event}{\term}) \otherwise \dobg}{\obj{\measures}_i}{\obj{\measures}_c}$ & 
                \multicolumn{2}{c}{\makecell{$\rightarrow$ $\blocked{(\prop \Rightarrow \within{\event}{\term})} {\obj{\measures}_i}{\obj{\measures}_c} \wedge \exists \obj{\measures}_v:\Class{\measures}$
                \\$(
                \obj{\measures}_v.\textit{time} = \obj{\measures}_i.time + T*(\term, \obj{\measures}_i)
                )
                \wedge\blocked{ \dobg}{\obj{\measures}_v}{\obj{\measures}_c} $ 
                }
                } \\
                \hline

                $\blocked{\prop \Rightarrow \obg}{\obj{\measures}_i}{\obj{\measures}_c}$ & 
                \multicolumn{2}{c}{$\rightarrow T^*(\prop, \obj{\measures}_i) \wedge \blocked{ \obg}{\obj{\measures}_i}{\obj{\measures}_c}$} \\
                \hline
                
                $\blocked{\obg}{\obj{\measures}_i}{\obj{\measures}_c}$ & 
                \multicolumn{2}{c}{$\rightarrow \exists \obj{\textit{block}_{\obg}}:\Class{\textit{block}_{\obg}}(\obj{\textit{block}_{\obg}}.\textit{i} = \obj{\measures}_i.\textit{time} \wedge \obj{\textit{block}_{\obg}}.\textit{c} = \obj{\measures}_c.\textit{time}) $} \\
                    \hline
                $\textit{axiomBlock}_{\obg}$ & \multicolumn{2}{c}{\makecell{$\rightarrow$ $\forall \obj{\textit{block}_{\obg}}:\Class{\textit{block}_{\obg}}\exists \obj{\measures}_i,\obj{\measures}_c:\Class{\measures} $ \\ $(\obj{\measures}_i.\textit{time} = \obj{\textit{block}_{\obg}}.i \wedge \obj{\measures}_c.\textit{time} = \obj{\textit{block}_{\obg}}.c \wedge  \intblocked{\obg} {\obj{\measures}_i} {\obj{\measures_c}})$}}\\
                \hline
                $\intblocked{\within{\event}{\term}} {\obj{\measures}_i} {\obj{\measures_c}}$ & \multicolumn{2}{c}{\makecell{$ \rightarrow \act{\within{\event}{\term}}{ \obj{\measures}_i}{\obj{\measures_c}} \wedge \bigvee_{\obg \in \obgbyhead{\neg \event}} (\exists \obj{\measures}_1:\Class{\measures}(\obj{\measures}_1.\textit{time} \le \obj{\measures}.time) $ \\ $\wedge \forced{\obg}{\obj{\measures}_1}{\measures_c} \wedge (\obj{\measures}.\textit{time} + 
                T^{*}(\term, \obj{\measures}) \ge \obj{\measures}_1.\textit{time} + 
                T^{*}(\term_1, \obj{\measures}))$ where $\term_1$ is $\obg$'s time limit }}\\
                \hline
                $\intblocked{\within{\neg \event}{\term}} {\obj{\measures}_i} {\obj{\measures_c}}$ & \multicolumn{2}{c}{\makecell{$ \rightarrow \act{\within{\neg \event}{\term}}{ \obj{\measures}_i}{\obj{\measures_c}} \wedge \bigvee_{\obg \in \obgbyhead{\event}} (\exists \obj{\measures}_1:\Class{\measures}(\obj{\measures}_1.\textit{time} \le \obj{\measures}.time) $ \\ $\wedge \forced{\obg}{\obj{\measures}_1}{\measures_c} \wedge (\obj{\measures}.\textit{time} + 
                T^{*}(\term, \obj{\measures}) \le \obj{\measures}_1.\textit{time} + 
                T^{*}(\term_1, \obj{\measures}))$ where $\term_1$ is $\obg$'s time limit }}\\
\hline
                
                $\act{\obg}{ \obj{\measures}_i}{\obj{\measures}_c}$ & \multicolumn{2}{c}{$\rightarrow$ $\triggered{\obg}{ \obj{\measures}_i} \wedge \neg \violated{\obg}{ \obj{\measures}_i}{\obj{\measures}_c}  \wedge \neg \fullfilled{\obg} {\obj{\measures}_i}{ \obj{\measures}_c} $}\\
\hline
                $\fullfilled{\within{\event}{\term},  \obj{\measures}_i, \obj{\measures}_c}$&
                \multicolumn{2}{c}{$\rightarrow$ $T^{*} (\within{\event}{\textsc{Min}(\term, \obj{\measures}_{c}.\textit{time})), \obj{\measures}}$  }\\
\hline
                $\fullfilled{\within{\neg \event}{\term}}{  \obj{\measures}_i}{ \obj{\measures}_c}$&
                \multicolumn{2}{c}{$\rightarrow$ $\timedtranslatestar{\within{\neg \event}{\term},  \obj{\measures}, \obj{\measures}_{c}.\textit{time}}$}\\
\hline

                $\violated{\obg}{ \obj{\measures}_i} {\obj{\measures}_c}$ &
                \multicolumn{2}{c}{$\rightarrow$ $\fullfilled{\noncompliant{\obg}, \obj{\measures}_i, \obj{\measures}_{c}.\textit{time}}$  }\\
\hline
               
                $\triggered{\obg}{\obj{\measures}}$ & \multicolumn{2}{c}{\makecell{$\rightarrow$ let $\textsc{trigger\_rule}(\obg) = \rulesyntax{\event \wedge \prop}{\dobg}$ where $(\prop_m \Rightarrow \obg) = \dobg[m]$ \\ if $m=1$ then $\exists \obj{\event}:\Class{\event}(\obj{\event}.\textit{time} = \obj{\measures}.\textit{time} \wedge T^*(\prop \wedge \prop_m, \obj{\measures}))$ \\
                else $\exists \obj{\measures}_i:\Class{\measures}(
                \violated{\dobg[m-1]}{ \obj{\measures}_i}{ \obj{\measures}} \vee \blocked{\dobg[m-1]}{\obj{\measures}_i}{\obj{\measures}})
                )
                $}}\\
\hline

                $\forced{\prop \Rightarrow \obg}{\obj{\measures}_i}{\obj{\measures}_c}$ where $\cobg = \dobg[m]$ & \multicolumn{2}{c}{\makecell{$\rightarrow$ $T^*(\prop, \obj{\measures}_i) \wedge \act{\obg}{\obj{\measures}_i}{\obj{\measures}_c} \wedge \neg (\blocked{\obg}{\obj{\measures}_i}{\obj{\measures}_c}) \wedge  $ \\
                $\exists \obj{\measures}_v:\Class{\measures}(\textbf{Violation}(\obg, \obj{\measures}_i,
                \obj{\measures}_v) \wedge \blocked{\dobg[m+1:]}{\obj{\measures}_v}{\obj{\measures}_c} )$ }}\\
                
                \hline
                \hline
                $\obgbyhead{h}$ & \multicolumn{2}{c}{$= \{ \obg  \mid (\within{h}{\term}) \text{ in the rule set} \}$} \\
               $\textsc{trigger\_rule}(\obg) = \rulesyntax{\event}{\dobg}$ & \multicolumn{2}{c}{if and only if $\obg \in \dobg$} \\ 
                \hline
                \hline
                
                \end{tabular}  
        }
\caption{\small The \fol encoding $\sctranslate{\srule, \ruleset}$  that 
describes a situation where the rule $\srule \in \ruleset$ to be situational conflicting. The table contains three parts: (1) the top-level encoding that describes the existence of a situation where $\srule$ is triggered at the last state ($\obj{\measures}$); (2) the \fol constraint for describing the situation is non-violating (i.e., $\timedtranslate{\ruleset, \obj{\measures}.\textit{time}})$); and (3) the \fol encoding for describing blocked obligations (i.e., $\blocked{\dobg} {\obj{\measures}_i}{\obj{\measures}_c}$) as well as another status of obligations where ${\obj{\measures}}_i$ is the state when the obligation is triggered and ${\obj{\measures}}_c$ is the last state of the situation.} 
\label{tab:sctranslate}
\end{figure*}

\begin{theorem}[Encoding of Situational Conflict]\label{thm:condconfencoding}
Let a rule set $\ruleset$ be given. For every rule $\srule \in \ruleset$, if the \fol formula $\sctranslate{\srule, \ruleset}$ is satisfiable, then $\srule$ is situationally conflicting.
\end{theorem}

\begin{sop}
If $(\domain, v)$ is a satisfying solution to $\sctranslate{\srule, \ruleset}$, then we use the same  method in the proof of Thm.~\ref{thm:transcorrectness} to construct a situation $\fos$ from $(\domain, v)$. We then show that the encoding in Tab.~\ref{tab:sctranslate} conforms with the semantics of non-violation and the status of obligations (in Def.~\ref{def:triggered}). Finally, we show that the constructed $\fos$ satisfies the sufficient condition for situational conflict (Lemma.~\ref{lemma:sc}), and thus $\srule$ is situational conflicting w.r.t $\fos$. \qed
\end{sop}

\section{Derivation Tree and proof reduction}
\label{ap:treereduction}
Given a refutation proof, one can construct a \textit{derivation graph} where every lemma is a node and its dependencies are the incoming edges to the node. The roots of the derivation graph are the input formulas and axioms (e.g., $\textit{axiom}_{mc}$) and the (only) leaf of the graph is the derived $\bot$. 


Using the derivation graph, one can check the soundness of the proof and reduce it by traversing the graph backwards from the leaf (step 8 in Fig.~\ref{fig:simpleproofexample}). While visiting a node, we first check if the lemma represented by the node can be soundly derived using the derivation rule with the lemmas in its dependency, and then reduce the dependency if not every lemma is necessary. Only the nodes representing the lemmas in the reduced dependencies are scheduled to be visited in the future. For instance, after checking the derivation step 7, its dependencies, 5 and 8, are scheduled to be visited next. If every scheduled node is visited without any failure, the proof is successfully verified, and the visited portion of the graph constitutes \emph{the reduced proof}. In a reduced proof, every derived lemma is used to derive $\bot$. 

We use two special derivation rules, \textbf{Input} and \textbf{Implication}.  
The rule \textbf{Input} adds an input \fol formula as a fact to the proof. 
The rule \textbf{Implication} derives new QF lemmas via logical implication, and it can be verified using an SMT solver by solving the formula $deps \wedge \neg C$ where $C$ is the derived lemmas and $deps$ are lemmas in the implication step's dependencies. 
The derivation is \emph{valid} if and only if the formula is UNSAT, and the UNSAT core returned by the SMT solver becomes the reduced dependencies.  

\begin{example}[Reduced Dependencies]\label{example:verifyimp}
    Let $L1: A > B$, $L2:  B > C$ and $L3:  C > 5$ be  three (derived) lemmas. Suppose a lemma $L4: A > C$ is derived using the \textbf{implication} rule given the dependencies $L1$, $L2$, $L3$. The rule can be verified by checking the satisfiability of $L1 \wedge L2 \wedge L3 \wedge \neg L4$. The result is UNSAT with an UNSAT core $L1$, $L2$ and $\neg L4$. Therefore, the reduced dependencies are $L1$ and $L2$.
\end{example}

\section{Condition for involved atomic elements}\label{ap:iae}

\begin{definition}[Involved atomic element]
Let a proof $L$ be given. We denote $\textit{Imp}(L)$ as the set of QF lemmas derived via or listed as dependencies for the derivation rule \textbf{Impl}.
An atomic proposition $\aprop$ is \emph{involved} if $\textit{Imp}(L)$ contains the quantifier-free (QF) formula $T^*(\aprop, \obj{\measures})$ for some relational object of class $\Class{\measures}$ where $T^*$ is the translation function for \dsl element defined in Tab.~\ref{tab:translate}. A triggering event $\event$ for \quoted{$\rulesyntax{\event \wedge \prop}{\dobg}$} is involved if $\textit{Imp}(L)$ contains a QF formula $\neg (\obj{\event}.\textit{ext}) \vee \neg T^*(\prop, \obj{\measures}) \vee F$ where $F$ is a QF formula, and both $\obj{\event}$ and $\obj{\measures}$ are relational objects of class $\Class{\event}$ and $\Class{\measures}$, respectively. Similarly, a triggering event $\event$ for \quoted{$\factsyntax{\event \wedge \prop}{\dobg}$} is involved if $Imp(L)$ contains a formula $\obj{\event}.\textit{ext} \wedge T^*(\prop, \obj{\measures}) \wedge F$. An obligation head $\event$ for $\within{\event}{\term}$ is involved if $Imp(L)$ contains a QF lemma $\obj{\event}.\textit{ext} \wedge \obj{\event}.\textit{time} \ge \obj{\measures}.\textit{time} \wedge F$ 
for some object $\obj{\event}$ and $\obj{\measures}$. An obligation deadline $\term$ for $\within{\event}{\term}$ is involved if $Imp(L)$ contains a quantifier-free lemma $\obj{\event}.\textit{time} \le \obj{\measures}.\textit{time} + T^*(\term, \obj{\measures}) \wedge F$ 
for some object $\obj{\event}$ and $\obj{\measures}$.
    
\end{definition}

\section{Additional Evaluation Results} \label{ap:eval}

To ensure correctness, we augmented the proof of absence of vacuous conflicts produced by N-Tool by integrating feedback from AutoCheck, an existing tool for analyzing vacuous conflicts. The results can be seen in Tbl.~\ref{tab:autoCheckVacous}.

 \begin{table}[h]
     \centering
     \begin{tabular}{c c c}
        \toprule
          case studies &  T-Tool & AutoCheck \\ \toprule
          ALMI &  0 & 0 \\
          \cellcolor{gray!10} ASPEN & \cellcolor{gray!10}0 & \cellcolor{gray!10}0 \\
          AutoCar &  0 & 0 \\
          BSN \cellcolor{gray!10}&  \cellcolor{gray!10}0 & \cellcolor{gray!10}0 \\
          DressAssist &  0 & 0 \\
          \cellcolor{gray!10}CSI-Cobot &  \cellcolor{gray!10}0 & \cellcolor{gray!10}0 \\
          DAISY &  0 & 0 \\
          \cellcolor{gray!10}DPA &  \cellcolor{gray!10}0 & \cellcolor{gray!10}0 \\
          SafeSCAD &  0 & 0 \\
        \bottomrule
     \end{tabular}
     \caption{Vacuous conflicts identified by N-Tool compared to AutoCheck}
     \label{tab:autoCheckVacous}
 \end{table}

\end{document}